\documentclass[aps,pra,twocolumn,longbibliography,superscriptaddress]{revtex4-1}
													
\usepackage{amsmath}
\usepackage{amsthm}
\usepackage{latexsym}
\usepackage{amssymb}
\usepackage{bm}
\usepackage{bbm}
\usepackage{appendix}
\usepackage{graphics,epstopdf}
\usepackage{physics}
\usepackage{color}
\usepackage[colorlinks=true,linkcolor=blue,citecolor=red,plainpages=false,pdfpagelabels]{hyperref}
\usepackage[normalem]{ulem}
\usepackage{newlfont}
\usepackage{amsfonts}
\usepackage{amsthm}
\usepackage{graphicx}
\usepackage{epsfig}
\usepackage{mathrsfs}
\usepackage[thinc]{esdiff}
\usepackage{caption}
\usepackage{subcaption}

\usepackage{times}

\newtheorem{proposition}{Proposition}
\newtheorem{corollary}{Corollary}
\def\tr{\operatorname{tr}}
\newtheorem{theorem}{Theorem}

\DeclareOldFontCommand{\rm}{\normalfont\rmfamily}{\mathrm}

\begin{document}

\title{Speed limits on correlations in bipartite quantum systems}
\author{Vivek Pandey}\email{vivekpandey@hri.res.in}

\affiliation{Harish-Chandra Research Institute,\\  A CI of Homi Bhabha National
Institute, Chhatnag Road, Jhunsi, Prayagraj 211019, India
}

\author{Divyansh Shrimali}\email{divyanshshrimali@hri.res.in}

\affiliation{Harish-Chandra Research Institute,\\  A CI of Homi Bhabha National
Institute, Chhatnag Road, Jhunsi, Prayagraj 211019, India
}

\author{Brij Mohan}\email{brijmohan@iisermohali.ac.in}

\affiliation{Harish-Chandra Research Institute,\\  A CI of Homi Bhabha National
Institute, Chhatnag Road, Jhunsi, Prayagraj 211019, India
}
\affiliation{Department of Physical Sciences, Indian Institute of Science Education and Research (IISER), Mohali-140306, India}

\author{Siddhartha Das}\email{das.seed@iiit.ac.in}
\affiliation{Center for Security, Theory and Algorithmic Research (CSTAR), International Institute of Information Technology, Hyderabad, Gachibowli, Telangana 500032, India}

\author{Arun Kumar Pati}
\email{akpati@iiit.ac.in}
\affiliation{Centre for Quantum Science and Technology (CQST), International Institute of Information Technology, Hyderabad, Gachibowli, Telangana 500032, India}

\begin{abstract}
Quantum speed limit is bound on the minimum time a quantum system requires to evolve from an initial state to final state under a given dynamical process. It sheds light on how fast a desired state transformation can take place which is pertinent for design and control of quantum technologies. In this paper, we derive speed limits on correlations such as entanglement, Bell-CHSH correlation, and quantum mutual information of quantum systems evolving under dynamical processes. Our main result is speed limit on an entanglement monotone called negativity which holds for arbitrary dimensional bipartite quantum systems and processes. Another entanglement monotone which we consider is the concurrence. To illustrate efficacy of our speed limits, we analytically and numerically compute the speed limits on the negativity, concurrence, and Bell-CHSH correlation for various quantum processes of practical interest. We are able to show that for practical examples we have considered, some of the speed limits we derived are actually attainable and hence these bounds can be considered to be tight.
\end{abstract}

\maketitle
\section{Introduction}

The quantum speed limit (QSL) is a fundamental limit imposed by quantum mechanics on the rate at which any quantum system evolves under a given dynamical process~\cite{Mandelstam1945,Margolus1998,Anandan1990,OO18,SCMD18}. It provides bound on the minimal time required to transport a quantum system from its initial state to a final state under a given dynamical process. Determination of quantum speed limits are pertinent for the design and realization of quantum technologies, e.g., quantum computing~\cite{AGN12,Aifer2022,Mohan2022}, quantum metrology~\cite{Campbell2018}, optimal control theory~\cite{Caneva2009,Campbell2017}, quantum thermodynamics~\cite{Das2018,Mukhopadhyay2018,F.Campaioli2018,Mohan2021}, etc. As quantum correlations are critical aspects of quantum theory from both fundamental and applied aspects, it is only natural to explore speed limits on quantum correlations. Entanglement lies at the heart of quantum theory as there is no classical counterpart to it~\cite{EPR1935,Bell1964}. Entanglement has proven to be a resourceful quantum correlation for several information processing tasks, e.g., quantum communication~\cite{Wiesner1992}, quantum cryptography~\cite{Ekert1992}, quantum computation~\cite{Jozsa1997,Jozsa2003}, quantum random number generators~\cite{CR12}, quantum metrology~\cite{Jonathon2008}, etc. However, there are fundamental limitations on the entangling abilities of bipartite and many body quantum interactions (cf.~\cite{Dur2001,Acoleyen2013,Marien2016,ADD15,Das2020,S_Das2021}).

Seminal works in Refs.~\cite{Mandelstam1945,Margolus1998,Anandan1990} have led to the current advancements of QSL and better understanding of its applications. QSL now is applied to wide range of topics in quantum information theory, e.g.,~\cite{S.Deffner2017,ST20,BDLR21,del21}. It has been extensively studied for closed system dynamics~\cite{Mandelstam1945,Margolus1998,Levitin2009,Anandan1990,Gislason1956,Eberly1973,Bauer1978,Bhattacharyya1983,Leubner1985,Vaidman1992,Uhlmann1992,Uffink1993,Pfeifer1995,Horesh1998,Soderholm1999,Giovannetti2004,Andrecut2004,Gray2005,Luo2005,Batle2005,Borras2006,Zielinski2006,Zander2007,Andrews2007,Kupferman2008,Yurtsever2010,Fu2010,Jones2010,Chau2010,Zwierz2012,S.Deffner2013,Fung2013,Poggi2013,Fung2014,Andersson2014,D.Mondal2016,Mondal2016,S.Deffner2017,Campaioli2018,Dimpi2022} and many progresses also have been recently made for open quantum dynamics~\cite{Deffner2013,Campo2013,Taddei2013,Pires2016,Jing2016,S.Deffner2020}. It has been found  that, for certain classes of states, quantum entanglement enhances the speed of evolution of composite quantum systems~\cite{Giovannetti2003}. There are progresses made in the direction to derive bounds on the maximal rates at which any Hamiltonian interaction can generate entanglement in bipartite and multipartite quantum systems~\cite{Marien2016,Vershynina2015,S_Das2021}. In Ref.~\cite{Camaioli2020}, speed limits on quantum resources were derived to study how quickly these resources can be generated or degraded by physical processes. In Refs.~\cite{Bera2013,Rudnicki20201}, QSL on entanglement for unitary dynamics has been derived using a geometric measure of entanglement. Whereas, quantum speed limits for entanglement and quantum discord have been derived by using distance based measure in Ref.~\cite{Paulson2022}. These speed limits on entanglement are often based on entanglement monotones which require optimization over the set of separable states~\cite{Camaioli2020,Paulson2022}.

In this paper, we derive speed limits on some of the widely discussed correlations, namely entanglement, quantum mutual information, and Bell-Clauser-Horne-Shimony-Holt (Bell-CHSH) correlations in bipartite quantum systems undergoing arbitrary dynamical processes. To derive speed limits on entanglement, we consider entanglement monotones like the negativity~\cite{Peres1996,Vidal2002} and concurrence~\cite{Wootters1998,Wootters2001,Rungta2001} which are comparatively easier to compute. The speed limit on the negativity we derive is applicable for arbitrary systems and dynamical processes, i.e., for both discrete and continuous variable quantum systems and processes. One of the major interest in calculation of the negativity of a bipartite state is that it is a necessary criterion for the state to be entanglement distillable~\cite{Horodecki1996, Rai01}. Our speed limits are lower bounds on the minimal time required for the changes in the negativity, concurrence, entropy, Bell-CHSH correlation, and quantum mutual information of bipartite quantum systems undergoing time-evolution. Here the time-evolution of a quantum system from its initial state to the final state is depicted by dynamical process describable by a completely positive trace preserving (CPTP) map (see e.g.,~\cite{Caruso14,Buscemi2016}), which is also called quantum channel. However, for an interval in between of initial and final states the process need not be CPTP~\cite{Jordan04,Rivas2012}. Therefore, speed limits that we derive are in general applicable for dynamics with or without memory-effect of environment (or bath)~\cite{An07,Caruso14,Das2018}. Some of these speed limits are tight and hence can be attainable for some quantum processes. We illustrate efficacy of our speed limits by applying them to various dynamical processes of practical interests~\cite{Dur2001,Lidar2019,Carrega2020,Weiss2006,Jing2013,Schindler2013,Childs2001,Mohseni2006,Maziero2010}, e.g., unitary process due to nonlocal Hamiltonian and open quantum dynamics due to pure dephasing process, depolarizing process, and amplitude damping process.

The organization of this paper is as follows. In Section~\ref{sec:prem}, we discuss the preliminaries and background required to arrive at the main results of this paper. In Section~\ref{sec:QSL}, we obtain speed limits on the negativity for arbitrary dynamics. In Section~\ref{sec:QSL-2}, we obtain speed limits on the concurrence and I-concurrence for unitary dynamics. In Section~\ref{sec:QSL-3}, we obtain speed limits on other correlations such as Bell-CHSH correlation, quantum mutual information, and entropy. In Section~\ref{sec:Examples}, we have analytically and numerically computed obtained speed limits for quantum systems evolving under some classes of interaction processes of practical interests. We also discuss tightness of these limits by considering some instances where the obtained bounds are attainable and some instances where they are loose. Finally, we provide concluding remarks in the last section.

\section{Preliminaries}\label{sec:prem}
In this section, we briefly review some of the standard notations and results common in the literature of quantum information theory.

Let $\mathcal{H}$ represents a separable Hilbert space with $\dim(\mathcal{H})$ that can be either finite ($<\infty$) or infinite ($=\infty$). Let $\cal{B}({\cal{H}})$ denote the algebra of bounded linear operators acting on $\cal{H}$ with $\mathbbm{1}_{\mathcal{H}}$ denoting the identity operator.  State of a quantum system is described by a density operator defined on the Hilbert space. Let $\cal{D}(\cal{H})$ denote the set of density operators acting on the Hilbert space $\cal{H}$. If $\rho \in\mathcal{D}(\mathcal{H})$ then it satisfies following properties: $\rho = \rho^{\dagger},\,\rho\geq 0,\,\tr(\rho)=1$. Let $\cal{H}_{A}$ denote the Hilbert space associated with a quantum system $A$. The state of quantum system $A$ is denoted  by density operator $\rho_{A} \in \cal{D}({\cal{H}}_{A})$. We denote by $\cal{H}_{AB}:=\cal{H}_{A}\otimes\cal{H}_{B}$ the Hilbert space of a bipartite  system $AB$. The state of bipartite system $AB$ is denoted by density operator $\rho_{AB}\in \cal{D}(\cal{H}_{AB}) $ and the density matrix of a subsystem, say $A$, which is given by the partial trace over the other subsystem $B$, i.e., $\rho_{A} = \tr_{\rm B}(\rho_{AB})$. A quantum system whose state is exactly known is said to be in pure state. The density operator representing a pure state $\psi\in\mathcal{D}(\mathcal{H})$ is a rank-one projection operator i.e. $\psi = \op{\psi}$, where $\ket{\psi}\in\mathcal{H}$. Otherwise, $\rho$ is in a mixed state and can be written as $\rho  =\sum_{i}p_{i}\psi_{i}$, where each $\psi_{i}$ is a pure density matrix with $\sum_{i}p_{i}=1$ and $p_{i}>0$ for all $i$.

The von-Neumann entropy of a quantum state $\rho_A \in \cal{D}(\cal{H}_A)$ is defined as:
 \begin{equation}
     S(A)_{\rho}:=S(\rho_A)= -\tr({\rho\ln{\rho}}).
 \end{equation}
 The von-Neumann entropy of a state $\rho_A$ is always non-negative and $S(A)_{\rho}=0$ if and only if the state $\rho_A$ is pure.
  
 The quantum relative entropy between any $\rho \in \cal{D}(\cal{H})$ and $\sigma\geq 0$ is defined as \cite{Umegaki1962}
 \begin{equation}
    D(\rho\Vert\sigma) := 
       \begin{cases}
                \tr(\rho(\ln{\rho}-\ln{\sigma})) &  \text{if}\  \operatorname{supp}(\rho)\subseteq \operatorname{supp}(\sigma), \\ 
                +\infty & \text{otherwise},
       \end{cases}
 \end{equation}
where $\operatorname{supp}(\rho)$ and $\operatorname{supp}(\sigma)$ are the supports of $\rho$ and $\sigma$, respectively.

The Schatten-p norm of an operator $O \in \cal{B}(\cal{H})$ is defined as:
\begin{equation}
       \norm{O}_{p}= \left(\tr \left|O\right|^{p}\right)^{1/p},
\end{equation}
where $\left|O\right|= \sqrt{O^{\dagger}O}$, $p\geq1,\ p\in\mathbb{R}$.  The operator norm, the Hilbert-Schmidt norm, and the trace norm corresponds to $p=\infty,2,1$ respectively and satisfy the inequality $\norm{A}_{\rm \infty} \leq \norm{A}_{\rm 2}   \leq \norm{A}_{\rm 1} $.

 \subsection{Quantum dynamics}
In the Schr\"odinger picture, the observables of the system are fixed and density operator evolves with time~\cite{Sud61}. The evolution of the density operator is governed by a linear, completely positive and  trace-preserving (CPTP) map called quantum channel. Although quantum channels provide a general way to describe the evolution of density operators for pure and mixed states both but sometimes it is convenient to use differential equations to describe the evolution instead of quantum channels. For example, evolution of pure states can be described by  unitary transformations which can easily be written in the form of differential equation (Schr\"{o}dinger equation). Under the assumption that evolution of density operator is smooth, the evolution of density operator can be written in the form of linear differential equation (called master equation):
 \begin{equation}
      \dot{\rho_t}:=\diff{\rho_t}{t}=\mathcal{L}_t(\rho_t) \label{Master_equation_density_operator},
 \end{equation}
where $\rho_t $ is the state of the system at time $t$ and $\mathcal{L}_t$ is the Liouvillian super-operator~\cite{Rivas2012} which in general can be time independent or time-dependent.

Let us  define initial and final states of time evolving quantum system as $\rho_{0}:=\rho_{t=0}$ and $\rho_{T}:=\rho_{t=T}$, respectively. We can drop the suffix $t$ if $\cal{L}_{t}$ is time independent and in this case the above evolution equation has the following formal solution:
\begin{equation}
    \rho_{t} = \operatorname{e}^{\cal{L}t}(\rho_{0}).
\end{equation}
Let $\Phi$ denote a linear CPTP map which maps density operator to density operators. The map $\Phi$ is called unital if $\Phi(\mathbbm{1}_{\cal{H}})=\mathbbm{1}_{\cal{H}}$. The adjoint map $\Phi^\dag:\cal{B}(\cal{H})\rightarrow\cal{B}(\cal{H})$ of $\Phi$ is a unique linear map that satisfies $\tr({\cal{O}\Phi(\rho)}) = \tr({\Phi^\dag}(\cal{O})\rho) $, $\forall \rho \in \cal{D}(\cal{H}), \cal{O} \in \cal{B}(\cal{H})$. The adjoint of a trace preserving map is unital.
 
 In the Heisenberg picture the density operator of the system is fixed and the observable evolves with time~\cite{Sud61}. The evolution is given by the adjoint (or dual) map ${\Phi^{\dag}}$  which keeps density operator fixed and takes an input observable to an output observable. Note that ${\Phi}^{\dag}$ may not be a quantum channel because in the case of observables, trace preserving condition is not necessary. In differential equation form the time evolution of an observable is given by adjoint-master equation:
 \begin{equation}
     \dot{\cal{O}}_t:= \diff{ \cal{O}_t }{t} =\cal{L}_{t}^{\dagger}(\cal{O}_t)\label{Master_equation_density_observable},
 \end{equation}
  where $\cal{O}_t $ is the observable of the system at time $t$ and  $\cal{L}^{\dagger}_{t}$ is adjoint of the Liouvillian super-operator. We define initial and final observable of time evolving quantum system as $\cal{O}_{0}:=\cal{O}_{t=0}$ and $\cal{O}_{T}:=\cal{O}_{t=T}$, respectively.

 \subsection{Correlations in bipartite quantum systems}
Quantum entanglement is a type of (non-classical) correlation found in bipartite and multipartite quantum systems. The concurrence~\cite{Wootters1998,Wootters2001}, the  negativity~\cite{Peres1996,Vidal2002} and the entanglement entropy~\cite{Bennett1996} are some of the widely discussed entanglement monotones. In the sequel, we define some other correlation quantifiers that we will be using for our purpose.

{\it Concurrence.---} The concurrence quantifies the entanglement present in a two-qubit quantum system. The square of concurrence of a two-qubit pure state $\psi_{AB}=\op{\psi_{AB}}$ is given as~\cite{Wootters1998,Wootters2001,Rungta2001}

\begin{equation}
    \mathscr{C}^{2}(\psi_{AB}) := {{\tr(\psi\mathcal{R}({\psi^{*}}))}},
\end{equation}
where, $\mathcal{R}(\psi):= (\sigma_{y}\otimes \sigma_{y}) \psi (\sigma_{y}\otimes \sigma_{y})$, $\sigma_y$ is the Pauli operator and $^{*}$ is the complex conjugation operation. The value of concurrence lies between 0 to 1 (i.e., $0\leq\mathscr{C}(\psi)\leq1$). For maximally entangled states (i.e., Bell states) $\mathscr{C}(\psi)=1$, while for product states $\mathscr{C}(\psi)=0$. 

The square of I-concurrence for a pure two-qudit state $\psi_{AB} =\op{\psi_{AB}}$ is given as~\cite{Rungta2001}
\begin{equation}
     \mathscr{C}_{I}^{2}(\psi_{AB}): = 2\nu_{d_{A}}\nu_{d_{B}}[ 1-\tr(\rho^2_{A})],
\end{equation}
where $\dim(\cal{H}_{A})=d_{A}$, $\dim(\cal{H}_{B})=d_{B}$, and $\nu_{d_{A}},\nu_{d_{B}}$ are constants. The I-concurrence reduces to concurrence for  $\nu_{d_{A}}=1$, and  $\nu_{d_{B}}=1$. The I-concurrence varies from 0 for pure product state to $\sqrt{2(d-1)/d}$, where $d=\min\{\dim(\cal{H}_{A}), \dim(\cal{H}_{B})\}$, for a maximally entangled state.

{\it Negativity.---} The negativity is an entanglement monotone to quantify the amount of entanglement present in an arbitrary bipartite quantum state. It is derived from the positive partial transpose (PPT) criterion for the separability of a bipartite quantum states. For an arbitrary bipartite quantum state $\rho_{AB}$, its negativity is defined as~\cite{Vidal2002}
\begin{equation}
\mathscr{N}(\rho) := \frac{\norm{\rho_{AB}^{\Gamma_B}}_{1}-1}{2},
\end{equation}
where $\rho_{AB}^{\Gamma_B}=\Gamma^{B}(\rho_{AB})$., and ${\Gamma_B}$ is partial transpose operation~(see Appendix \ref{pattial_transpose}). 

Note that, the negativity can be zero even when the state is entangled because the PPT criterion is only necessary but not sufficient for bipartite state to be separable. However, for two-qubit states (and also for bipartite state where one system is qutrit while the other is qubit), the set of PPT states and separable states coincide (are same). The negativity is a convex function i.e, $\mathscr{N}(\sum_{i}p_{i}\rho_{i})\leq\sum_{i}p_{i}\mathscr{N}(\rho_{i})$, where $\{p_{i}\}_i$ is probability distribution (hence, $\sum_{i}p_{i}=1$ and $p_i\geq 0$ for all $i$) and $\rho_{i}$'s are density operators. All states that are not PPT are called NPT states. NPT states are entangled. Therefore, the non-zero negativity implies that the given bipartite state is entangled.

{\it Entanglement entropy.----}
 The entanglement entropy is an entanglement monotone and quantifies the amount of entanglement present between subsystems $A$ and $B$ when the joint state of $AB$ is pure, i.e., $S(AB)_\psi=0$.  The entanglement entropy of composite system $AB$ in a pure state $\psi_{AB}$ is defined as~\cite{Bennett1996}
\begin{equation}
    E_S(\psi_{ AB}):=S(\rho_A)=S(\rho_B)= - \tr(\rho_{ A}\ln\rho_{ A}),
\end{equation}
where $\rho_{A}:=\tr_{\rm B}(\psi_{AB})$.

{\it Bell-CHSH observable.---} The non-classical correlation existing in a bipartite quantum state is detected by the violation of Bell-CHSH inequality. These correlations appear to be stronger than anything explained by classical physics. The violation of Bell-CHSH inequality for a given bipartite quantum state $\rho_{AB}$ can be checked by estimating the expectation value of the Bell observable $\cal{B}$. For a two-qubit systems, which may be represented as a pair of spin-1/2 particles, the Bell-CHSH observable has the following
general form~\cite{CHSH1969,Horodecki1995}
\begin{equation}
    \cal{B} := \hat{a}.\vec{\sigma} \otimes \left(\hat{b}+\hat{b'}\right).\vec{\sigma} + \hat{a'}.\vec{\sigma} \otimes \left(\hat{b}-\hat{b'}\right).\vec{\sigma},
\end{equation}
where $\hat{a}$, $\hat{a'}$, $\hat{b}$ and $\hat{b'}$ are unit vectors in a $3$-dimensional Euclidean space $\mathbb{R}^{3}$ and $\vec{\sigma}$ is the Pauli spin vector operator. The state $\rho_{AB}$ violates Bell-CHSH inequality for $\abs{\langle\cal{B}\rangle_\rho}> 2$. The optimal Bell measurement settings for Bell observable for any given bipartite state $\rho_{AB}$ has been derived in Ref.~\cite{Horodecki1995}. Any state that violates the Bell-CHSH inequality is said to be nonlocal as no local-realistic hidden variable models can depict such correlations.

{\it Quantum mutual information.---} The correlation present in bipartite quantum system can be quantified by quantum mutual information. The quantum mutual information of a bipartite state $\rho_{AB}$ is defined as
\begin{equation}
    {I}(A;B)_\rho:= S({A})_{\rho}+S({B})_{\rho}-S({AB})_{\rho},
\end{equation}
where $\rho_A$ and $\rho_B$ are reduced density matrices of $\rho_{AB}$. The quantum mutual information can also be stated in
terms of the quantum relative entropy as
  \begin{align}
  {I}(A;B)_\rho :=\min_{\omega_{A} \in \cal{D}(\cal{H}_{A}), \omega_{B}\in \cal{D}(\cal{H}_{B})}D(\rho_{AB}||\omega_{A}\otimes \omega_{B}).\label{equ:Quantum_Mutual_Information}
\end{align}
The quantum mutual information $I(A;B)$ of any bipartite  state $\rho_{AB}$ is a non-negative quantity i.e., $I(A;B)\geq0$ (where equality holds if and  only if the state is factorize i.e., a product state) and upper bounded by $2\ln\left({\rm min}\{\rm dim(\cal{H}_{A}),dim(\cal{H}_{B})\}\right)$.

\subsection{Speed limits on observable}

The speed limit on observables are defined as a bound on the maximum evolution speed of the expectation value of a given observable of a quantum system undergoing dynamical evolution, which might be unitary or non-unitary. It sets the lower bound on the evolution time of the quantum system needed to evolve between different expectation values of a given observable. The bound on evolution time of expectation value of an observable for an arbitrary dynamics reads as (see Appendix \ref{QSL:Observable})
\begin{equation}
 T \geq  T_{\rm OQSL}= \frac{\left| \langle\cal{O}_{T}\rangle_{\rho}-\langle\cal{O}_{0}\rangle_{\rho}\right|}{\norm{\rho}_{1}}\max\left\{\frac{1}{\Lambda_{T}^{\infty}},\frac{1}{\Lambda_{T}^{1}},\frac{1}{\Lambda_{T}^{2}}\right\},
 \end{equation}
where $\Lambda_{T}^{\alpha}=\frac{1}{T}\int_{0}^{T}{\rm d}t\norm{\cal{L}_{t}^{\dagger}(\cal{O}_{t})}_{\alpha}$ for $\alpha\in\{1,2,\infty\}$ is the evolution speed of the observable of the given system under the dynamics $\dot{\cal{O}}_t:= \diff{ \cal{O}_t }{t} =\cal{L}_{t}^{\dagger}(\cal{O}_t)$.

 Note that speed limits on observables have been previously derived in Refs.~\cite{B.Mohan2021,Pintos2021} using the Cauchy-Schwarz inequality. Our derivation employs the H\"older's inequality, and therefore, our bound is more general than the previous ones. 
In Refs.~\cite{B.Mohan2021,Carabba2022}, state-independent speed limits on observables are derived by considering the Hilbert-Schmidt inner product for observables. Furthermore, Ref.~\cite{Nikals2022} formulates the speed limit of super-operators, which are operators in the operator space, and showed some applications in many-body physics.

\section{Quantum Speed Limits}
 In general, quantum speed limits represent fundamental constraints imposed by the quantum theory on the evolution speed of quantum systems. Entangling abilities of quantum interactions are of wide interest from both fundamental and applied aspects, see e.g.,~\cite{S.Das2018,S_Das2021}. Quantum correlations are fundamental in the quantum information theory as they act as resources for several quantum information processing tasks. For instances, entanglement and nonlocal quantum correlations are useful properties (resources) for the tasks of teleportation, quantum key distribution, quantum communication, quantum sensing, etc.~\cite{Bennett1993, Ekert1992,CR12, Jonathon2008,S_Das2021}. 
 
 In this section, we discuss limitations on the minimal time taken for changes in some desirable correlation measures of bipartite quantum systems undergoing bipartite dynamical processes. Our main focuses are speed limits on the negativity and concurrence, which are entanglement monotones useful in the resource theories of entanglement, see e.g.,~\cite{Horodecki2009, Das_S_2018,BDWW19}. We also inspect speed limits on other correlations, namely Bell-CHSH observable $\cal{B}$ and quantum mutual information ${I}(A;B)_\rho$.
 
 We note that speed limits on entanglement for unitary dynamics has been studied earlier using geometric measure of entanglement~\cite{Bera2013,Rudnicki20201}. In Refs.~\cite{Camaioli2020}, speed limits on entanglement has been obtained using divergence based measure for open quantum dynamics. The bounds obtained in Refs.~\cite{Bera2013,Rudnicki20201,Camaioli2020,Paulson2022} are challenging to calculate in general as it requires optimization over all separable sates. Here, we have obtained speed limits on entanglement using the negativity and the concurrence, which is arguably easier to calculate. Moreover, speed limits on entanglement obtained using the negativity is applicable for arbitrary dynamics, i.e., for both closed and open dynamical processes. Speed limits obtained for negativity also provide limitations on the minimal time required for transitions of a PPT state to NPT (non-negative under partial transposition) and NPT state to PPT under any given dynamical processes.

\subsection{Speed limit on the negativity}\label{sec:QSL}
We now discuss the first main result of this work that provides a lower bound $T_{\rm NSL}$ on the time taken for the change in the value of negativity of a bipartite quantum system evolving under a given dynamical processes.
\begin{theorem}\label{thm:1}
 Consider any bipartite quantum system $AB$, where each of $\dim(\cal{H}_A)$ and $\dim(\cal{H}_B)$ can be either finite or infinite. The minimal time $T$ taken for the bipartite system to bring certain amount of change in its negativity by evolving under an arbitrary given quantum dynamical process with associated Liouvillian $\cal{L}_t$ is lower bounded by 
\begin{equation}\label{eq:theorem-1}
T \geq  T_{\rm NSL}= \frac{2\abs{\mathscr{N}(\rho_{T})-\mathscr{N}(\rho_{0})}}{ \Lambda^{N}_{T}},
\end{equation}
where $\rho_0$ is the initial state (at $t=0$), $\rho_T$ is the final state (at $t=T$), and $\Lambda^{N}_{T}:=\frac{1}{T}\int_{0}^{T}{\rm d}t\norm{\cal{L}_t(\rho^{T_{B}})}_{1}$. $\Lambda^N_T$ can be interpreted as the evolution speed of the negativity of the given system and process.
\end{theorem}
\begin{proof}
 Consider the evolution of bipartite state $\rho$ in the time interval interval $ \mathscr{I}:=[0, T]$.
 The negativity of time evolved bipartite state $\rho_{t}$ given by
\begin{equation}
\mathscr{N}(\rho_t)  = \frac{\norm{\rho^{T_{B}}_t}_{1}-1}{2}.
\end{equation}
  Now consider the $\epsilon$ neighbourhood of $t\in \mathscr{I}$ (i.e. an interval $(t-\epsilon, t+\epsilon)$ where  $\epsilon$ is a number arbitrarily close to zero), and the following difference
\begin{equation}
  \mathscr{N}(\rho_{t+\epsilon})-\mathscr{N}(\rho_t) = \frac{1}{2}\left(\tr\left|\rho^{T_{B}}_{t+\epsilon}  \right| - \tr\left|\rho^{T_{B}}_{t}\right|\right).
\end{equation}
We assume that evolution of density operator  $\rho$ is smooth, i.e., $\rho_{t}$ is differentiable at each $t \in \mathscr{I} $, which also implies that $\dot{\rho}^{T_{B}}_{t}$ and $\tr\left|\dot{\rho}^{T_{B}}_{t}\right|$ are  also well defined at each $t \in \mathscr{I}$.  We further assume that negativity is differentiable in $\mathscr{I}$, so the  left hand derivative and the right hand derivative of negativity must be equal at each point in $\mathscr{I}$ and also equal to the derivative of negativity. Now, multiplying by $\frac{1}{\epsilon}$ on both the sides of above equation and taking limit $\epsilon\to 0$, we obtain
\begin{align}
     \frac{{\rm d} }{{\rm d}t}\mathscr{N}(\rho_t) & =\frac{1}{2}\lim_{\epsilon \to 0}\frac{\tr\left|\rho^{T_{B}}_{t+\epsilon}  \right| - \tr\left|\rho^{T_{B}}_{t}\right|}{\epsilon} \label{equ:negativity_derivative}.
\end{align}
Using the Taylor expansion we have: \\
\begin{equation}
    \left| \rho^{T_{B}}_{t+\epsilon}  \right| = \left| \rho^{T_{B}}_{t}+\epsilon \dot{\rho}^{T_{B}}_{t} + o(\epsilon^2)  \right|, \label{equ:taylor_expansion}
\end{equation}
Let us now take the absolute value on both the sides of Eq.~\eqref{equ:negativity_derivative} and use Eq.~\eqref{equ:taylor_expansion} to calculate the limit, we then get
\begin{align}
      \left|\frac{{\rm d} }{{\rm d}t}\mathscr{N}(\rho_t)\right| &=\frac{1}{2}\left|\lim_{\epsilon \to 0}\frac{ \tr\left| \rho^{T_{B}}_{t}+\epsilon \dot{\rho}^{T_{B}}_{t} + o(\epsilon^2)  \right| - \tr\left|\rho^{T_{B}}_{t}\right|}{\epsilon}\right|.
\end{align}      
Now, we leave the terms of $o(\epsilon^2)$ in the Taylor expansion of $\rho^{T^{B}}_{t+\epsilon}$ and use the triangular inequality $\tr\left| A+B \right|\leq\tr\left| A \right| + \tr\left| B \right| $ to further simplify above equation. We then obtain      
\begin{align}
  \left|\frac{{\rm d} }{{\rm d}t}\mathscr{N}(\rho_t)\right|    &\leq \frac{1}{2}\left|\lim_{\epsilon \to 0}\frac{\tr \left| \rho^{T_{B}}_{t}\right|+\epsilon\tr\left| \dot{\rho}^{T_{B}}_{t} \right| - \tr\left|\rho^{T_{B}}_{t}\right|}{\epsilon}\right|\nonumber\\
      &=\frac{1}{2}\tr\left({\left|\dot{\rho}^{T_{B}}_t\right|}\right)=\frac{1}{2}\norm{\dot{\rho}^{T_{B}}_t}_{1} \label{equ:23}.
\end{align}

The above inequality~\eqref{equ:23} is the upper bound on that the rate of change of the negativity of the quantum system evolving under given dynamics. After integrating the above equation with respect to time $t$, we obtain 

\begin{align}
  \int_{0}^{T} {\rm d}t\left|\frac{{\rm d} }{{\rm d}t}\mathscr{N}(\rho_t)\right|& \leq \frac{1}{2}\int_{0}^{T}{\rm d}t\norm{\dot{\rho}^{T_{B}}_t}_{1} \nonumber\\
  &= \frac{1}{2}\int_{0}^{T}{\rm d}t\norm{\cal{L}_t(\rho^{T_{B}})}_{1}.
\end{align}
From the above inequality, we get the desired bound:
\begin{equation}
T \geq   \frac{2\left| \mathscr{N}(\rho_{T})-\mathscr{N}(\rho_{0})\right|}{ \Lambda^{N}_{T}}.
\end{equation}
We also provide an alternative proof of Theorem~\ref{thm:2} (see Appendix~\ref{alterproof}).
\end{proof}

 The bound~\eqref{eq:theorem-1} holds for both the generation and degradation of entanglement due to quantum dynamical processes. The negativity is monotone under local operation and classical communication (LOCC). Therefore, whenever bipartite dynamical process from initial time to final time can be represented as an LOCC map, we have $\mathscr{N}(\rho_{T})\leq \mathscr{N}(\rho_{0})$; see also Remark 4.4 of Ref.~\cite{S.Das2018} in this context. It is a trivial observation that that $T_{\rm NSL}=0$ if and only if there is no difference between the negativity between initial and final states, where we assume that the dynamics is such that $\Lambda_T$ is finite. There are multiple scenarios under which no change in the negativity and hence $T_{\rm NSL}=0$ may occur. For instances, $(i)$ if the initial state is a fixed point of the dynamical process, $(ii)$ if the dynamical process is PPT-preserving map, i.e., processes that map PPT states to PPT states, $(iii)$ if we are choosing evolution duration of the state under dynamical process such that negativity at the initial and final time points are same.

\subsection{Speed limit on the concurrence}\label{sec:QSL-2}
The concurrence was first introduced for pure two-qubit states in Ref.~\cite{Wootters1998} and later a generalized version of concurrence called I-concurrence for pure two-qudit states was introduced in Ref.~\cite{Rungta2001}. We now discuss lower bounds on the minimal time for the certain amount of change in the concurrence and the I-concurrence for two-qubit and two-qudit systems, respectively, evolving under time-dependent Hamiltonians $H_t$. For a time-dependent Hamiltonian $H_t$, subscript $t$ is to denote time-point $t$.
\begin{theorem}\label{thm-t-independent}
Consider a closed two-qubit quantum system $AB$ which is in a pure state. The minimal time $T$ taken for the (closed) system to evolve for a certain amount of change in the square of its concurrence under unitary dynamics generated by a time-dependent Hamiltonian $H_{t}$ is lower bounded by 
\begin{equation}
T \geq  T_{\rm CSL}= \frac{\hbar}{4}\frac{\abs{ \mathscr{C}^2(\psi_T)-{ \mathscr{C}^2(\psi_0)}}}{\Lambda^{C}_{T}},
\end{equation}
where $\psi_t= \mathcal{T}\exp(\int_{0}^{t}\mathcal{L}_{t'}{\rm d}t')\psi_0$, with $\psi_0$ denoting the initial state (at $t=0$), $\psi_T$ denoting the final state (at $t=T$), $\mathcal{T}$ is time ordering operator, and $\Lambda^{C}_{T}=\frac{1}{T}\int_{0}^{T}\sqrt{\tr(\psi_{t}H_{t}^2)}{\rm d}t$. $\Lambda^C_T$ can be interpreted as the evolution speed of the square of the concurrence of the given system and process.
\end{theorem}

See Appendix \ref{proof_of_tmh_2} for the detailed proof of the above theorem. An immediate consequence of the above theorem is the following corollary.
\begin{corollary}\label{thm:2}
Consider a closed two-qubit quantum system $AB$ which is in a pure state. The minimal time $T$ taken for the (closed) system to evolve for a certain amount of change in the square of its concurrence under unitary dynamics generated by a time-independent Hamiltonian $H$ is lower bounded by 
\begin{equation}\label{eq:theorem-2}
T \geq  T_{\rm CSL}= \frac{\hbar}{4}\frac{\abs{ \mathscr{C}^2(\psi_T)-{ \mathscr{C}^2(\psi_0)}}}{\sqrt{\tr(\psi_{0}H^2)}},
\end{equation}
where $\psi_t= \exp(\frac{\iota Ht}{\hbar})\psi_0\exp(\frac{-\iota Ht}{\hbar})$, with $\psi_0$ denoting the initial state (at $t=0$) and $\psi_T$ denoting the final state (at $t=T$). 
\end{corollary}

We now derive speed limits on the I-concurrence.
\begin{proposition}\label{prop:1}
Consider a finite-dimensional bipartite quantum system $AB$ initially in a pure state. The minimal time $T$ taken for the closed system to bring certain amount of change in the square of its I-concurrence under unitary dynamics generated by a time-dependent Hamiltonian $H_t$ is lower bounded by
\begin{equation}\label{eq:prop-1}
T \geq  T_{\rm ICSL}= \frac{\abs{{ \mathscr{C}_{I}^2(\psi_T)}-{ \mathscr{C}_{I}^2(\psi_0)}}}{\Lambda_{T}^{I}},
\end{equation}
where $\psi_0$ is the initial state (at $t=0$), $\psi_T$ is the final state (at $T=0)$, $\Lambda_{T}^{I}:=4\nu_{d_{A}}\nu_{d_{B}}\frac{1}{T}\int_{0}^{T}\norm{\rho^{A}_t}_{\rm 2}\norm{\tr_{\rm B}\left(\cal{L}_t(\psi_{t})\right)}_{\rm 2}{\rm d}t$ for the given unitary dynamics $\mathcal{L}_t(\psi_t)=-\frac{\iota}{\hbar}[\psi_t,H_t]$, and $\rho^A_t:=\tr_{\rm B}\psi_t$.
\end{proposition}

\begin{proof}
 The square of I-concurrence of bipartite state $\rho_{t}$ is given as \begin{equation}
  \mathscr{C}_{I}^{2}(\psi_{t}) = 2\nu_{d_{A}}\nu_{d_{B}}[ 1-\tr((\rho^A_t)^2)],
 \end{equation}
 where $\rho^A_t:=\tr_{\rm B}\psi_t$. After differentiating above equation with respect to time $t$, we then obtain
 \begin{align}
     \frac{\rm d}{{\rm d}t} \mathscr{C}_{I}^{2}(\psi_{t})
     &= -4\nu_{d_{A}}\nu_{d_{B}}\tr({\rho_t^{A}\dot{\rho}_t^{A}})\nonumber\\
     &= -4\nu_{d_{A}}\nu_{d_{B}}\tr(\rho_t^{A}\tr_{\rm B}(\cal{L}_t(\psi_{t})).
 \end{align}
Now let us take absolute value of above equation and applying Cauchy--Schwarz inequality, we then obtain
 \begin{align}
     \left|\frac{\rm d}{{\rm d}t} \mathscr{C}_{I}^{2}(\psi_{t}) \right|\leq 4\nu_{d_{A}}\nu_{d_{B}}\norm{\rho_t^{A}}_{\rm 2}\norm{\tr_{\rm B}\left(\cal{L}_t(\psi_{t})\right)}_{\rm 2}.
 \end{align}

 The above inequality is the upper bound on that the rate of change of square of the concurrence of the quantum system evolving under given dynamics.
After integrating above equation with respect to time $t$, we obtain
\begin{equation}
    \int_{0}^{T} \left|\frac{\rm d}{{\rm d}t} \mathscr{C}_{I}^{2}(\psi_{t}) \right|{\rm d}t \leq  \nu_{AB} \int_{0}^{T}\norm{\rho^{A}_t}_{\rm 2}\norm{\tr_{\rm B}\left(\cal{L}_t(\psi_{t})\right)}_{\rm 2} {\rm d}t,
\end{equation}
where $\nu_{AB}:= 4\nu_{d_{A}}\nu_{d_{B}}$. From the above inequality, we get the bound 
\begin{equation}
    T\geq  \frac{\left|{ \mathscr{C}_{I}^2(\psi_T)}-{ \mathscr{C}_{I}^2(\psi_0)}\right|}{\Lambda_{T}^{I}}.
\end{equation}
\end{proof}

The QSL on concurrence and I-concurrence applies to both entanglement generation and degradation processes. In particular, our bounds~\eqref{eq:theorem-2} and~\eqref{eq:prop-1} can determine the minimal time required to prepare a pure bipartite entangled state from a pure product state via unitary dynamics. It is a trivial observation that under local (i.e., separable) unitary dynamics concurrence and I-concurrence are invariant, therefore under such dynamics both $T_{\rm CSL}$ and $T_{\rm ICSL}$ are zero. Note that for any two-qubit system in a pure state the I-concurrence reduces to the concurrence for $\nu_{d_{A}}=1$ and $\nu_{d_{B}}=1$ but $T_{\rm CSL}$ and $T_{\rm ICSL}$ may not be equal because the evolution speed for a given dynamics could be different for both the bounds. For example, in the case of time-independent Hamiltonian, the evolution speed of concurrence is independent of the time interval while the evolution speed of I-concurrence depends on the time interval.

\subsection{Speed limits on other correlations}\label{sec:QSL-3}
Quantum mechanics allow for correlations between systems that cannot be depicted by any classical systems. It is known that no local-realistic hidden variable theories can predict all the outcomes exhibited by quantum correlations~\cite{EPR1935,Bell1964,Brunner2014} (see also Refs.~\cite{Bancal2013,Home2015}). Entangled states with no local-realistic hidden variable models are deemed nonlocal states. Bell-CHSH observables are used to test nonlocality of bipartite quantum states~\cite{CHSH1969,Hensen2015}. Apart from quantum corrrelations like entanglement and nonlocality, there is also interest in quantifying total amount of correlations between two systems. A quantifier of total amount of correlations present in arbitrary bipartite quantum system is quantum mutual information. It captures both the classical and truly quantum correlations present in a bipartite quantum system. In this section, we provide speed limits on the Bell-CHSH observable and the quantum mutual information for certain classes of quantum dynamics and speed limits on the von-Neumann entropy for arbitrary dynamics.

{\it Speed limit on Bell-CHSH correlation.---}\label{thm:3} In the Heisenberg picture, it is the operator which changes in time while the density operator remains fixed. For any dynamics of two-qubit quantum system with initial state $\rho$, the minimum time needed for the Bell-CHSH observable $\cal{B}_{t}$ to attain expectation value $\langle\cal{B}_{T}\rangle_{\rho}$, starting with the initial expectation value $\langle\cal{B}_{0}\rangle_{\rho}$, is lower bounded by (see Appendix~\ref{QSL:Observable})
\begin{equation}\label{eq:theorem-3}
 T \geq  T_{\rm BQSL}= \frac{| \langle\cal{B}_{T}\rangle_{\rho}-\langle\cal{B}_{0}\rangle_{\rho}|}{\norm{\rho}_{1}}\max\left\{\frac{1}{\Lambda_{T}^{\infty}},\frac{1}{\Lambda_{T}^{1}},\frac{1}{\Lambda_{T}^{2}}\right\},
\end{equation}
where $\Lambda_{T}^{\alpha}:=\frac{1}{T}\int_{0}^{T}{\rm d}t||\cal{L}_{t}^{\dagger}(\cal{B}_{t})||_{\alpha}$ for $\alpha\in\{1,2,\infty\}$  is the evolution speed of Bell-CHSH observable of the given system, $\langle\cal{B}_{0}\rangle_{\rho}$ and  $\langle\cal{B}_{T}\rangle_{\rho}$ are expectation value of Bell-CHSH observable at $t=0$ and $t=T$, respectively. In Appendix~\ref{app:chsh}, we also derive speed limit on Bell-CHSH correlation for bipartite quantum dynamics describable as separable maps.

We note that the sharpest bound is the operator norm-based bound. However, determining the Hilbert-Schmidt norm is comparatively easier to compute for general quantum dynamics. 

{\it Speed limit on quantum mutual information.---}
Consider quantum systems $A$ and $B$ which are of arbitrary dimensions. We assume that the systems are initially uncorrelated before they interact. Starting from a product state $\omega^{AB}_0:=\rho^{A}_0\otimes\rho^{B}_0$, the minimal time $T$ required to bring a certain amount of change in the quantum mutual information of the system evolving under an arbitrary quantum dynamics with time-dependent Liouvillian $\mathcal{L}_t$ is lower bounded by
\begin{equation}\label{eq:theorem}
T \geq  T_{\rm MISL}= \frac{ I(A;B)_{\omega_T}}{ \Lambda^{M}_{T}},
\end{equation}
where $\omega_t$ denotes the state at time $t$ with $\omega_0$ being the initial state (at $t=0$) and $\omega_T$ being the final state (at $t=T$) and $\Lambda^{M}_{T}:=\frac{1}{T}\int_{0}^{T}{\rm d}t\norm{\mathcal{L}_{t}({\omega_t})}_{ 2} \norm{\ln\omega_t-\ln{\omega_0})}_{2}$ (see Appendix~\ref{MI} for proof). Here we are implicitly assuming dynamics for which $\operatorname{supp}(\omega_{t_2})\subseteq \operatorname{supp}(\omega_{t_1})$ for all valid time points $t_2\geq t_1$~(see Appendix \ref{operator_function}). $\Lambda_T$ can be interpreted as the evolution speed of the quantum mutual information of the given system and process. This bound~\eqref{eq:theorem} only holds for quantum mutual information generation process. The bound~\eqref{eq:theorem} provides the minimal time required to prepare a bipartite correlated state from a uncorrelated state.

{\it Speed limit on the entropy.---}  Consider a quantum system $A$, where $\dim(\cal{H}_A)\leq \infty$.  The minimal time $T$ taken for the quantum system to bring certain amount of change in the entropy under an arbitrary given quantum dynamical process with associated Liouvillian $\cal{L}_{t}$  is lower bounded by
\begin{equation}\label{EEbound}
  T\geq T_{{\rm ESL}} = \frac{\left|S(\rho_T) -S(\rho_0)\right|}{ {\Lambda}^S_{T}},
\end{equation}
where $\rho_{t}$ denotes the state at time $t$ with $\rho_{0}$ being the initial state (at $t=0$) and $\rho_{T}$ being the final state (at $t=T$) and $\Lambda^S_{T}:=\frac{1}{T}\int_{0}^{T}{\rm d}t \norm{\mathcal{L}_{t}({\rho_t})}_{2} \norm{\ln \rho_t}_{2}$ (see proof in Appendix~\ref{entropy}).

Several entanglement measures such as the entanglement of formation~\cite{Bennett1996,Wootters2001}, the distillable entanglement~\cite{Bennett1996,C.Bennett1996}, and the relative entropy of entanglement~\cite{Vedral1997,Vedral1998} reduce to the entanglement entropy in the case of a closed bipartite system~\cite{Osborne2002}. Thus, in this case, the above bound~\eqref{EEbound} also sheds light on speed limits on entanglement for bipartite systems in a pure state (cf.~\cite{Das2018,S.Das2018}). The above bound~\eqref{EEbound} is valid for both the entropy generation and degradation processes, and also tighter than that obtained in the Theorem 1 of Ref.~\cite{Mohan2022}.

\section{Numerical results for some practical examples}\label{sec:Examples}

In this section, we apply our speed limits on the negativity, concurrence, and Bell-CHSH observable for some classes of quantum dynamics of wide practical interest~\cite{Rivas2012,Lidar2019}. In particular, we consider an unitary dynamics and some non-unitary processes classified as pure dephasing process, depolarising process, and amplitude damping process.

First, we explore speed limits on the negativity~\eqref{eq:theorem-1} and the concurrence~\eqref{eq:theorem-2} for unitary dynamics generated by non-local Hamiltonian.

 {\it Unitary process}.--- Let us consider two-qubit systems $AB$ interacting via a non-local Hamiltonian $H_{AB}$. Any two-qubit general Hamiltonian can always be expressed as
\begin{align}
    H_{AB} & = \sum_{i = \{x,y,z\}} \alpha_{i}\sigma^A_{i}\otimes \mathbbm{1}^B + \sum_{j = \{x,y,z\}} \mathbbm{1}^A \otimes  \beta_{i}\sigma^B_{i} \nonumber \\
    & \hspace{0.5cm} + \sum_{i,j = \{x,y,z\}} \gamma_{i,j} \sigma^A_{i} \otimes \sigma^B_{j} ,
\end{align}
where $\vec{\alpha} \in \mathbb{R}^3$, $\vec{\beta} \in \mathbb{R}^3$, $\gamma$ is a $3\times3$ real matrix and $\sigma^{A}_{i}$ and $\sigma^{B}_{i}$ are Pauli operators acting on $A$ and $B$, respectively.

We can always  perform local unitary operations without changing the amount of entanglement present in the system. Here we are only interested in the entanglement dynamics which allows us to restrict the form of the Hamiltonian to those which can be written in the following form~\cite{Dur2001} 
\begin{equation}
    \widetilde{H}_{AB}^{\pm} =  \mu_{x}  \sigma^A_{x}\otimes\sigma^B_{x} \pm \mu_{y}  \sigma^A_{y}\otimes\sigma^B_{y} +\mu_{z}  \sigma^A_{z}\otimes\sigma^B_{z},
\end{equation}
where $\mu_{x}$, $\mu_{y}$ and $\mu_{z}$ are singular values of matrix $\gamma$ with ordering $\mu_{x} \ge \mu_{y} \ge \mu_{z} \ge 0$.  In unitary dynamics Eq.~\eqref{Master_equation_density_operator} reduces to Liouville-von Neumann equation: \begin{equation}
     \dot{\psi_t}= -\iota[H,\psi_{t}],
\end{equation}
where $H$ is the Hamiltonian of the system and we have taken $\hbar=1$. We take $\widetilde{H}_{AB}^{+}$ as system's Hamiltonian without loss of generality~\cite{Dur2001} and $\psi_0$ as initial state with $\ket{\psi_{0}} = \sqrt{p}\ket{00}+\sqrt{1-p}\ket{11}$ for $p\in[0,1]$. Note that $\psi_0$ for $p=1/2$ is a fixed point for the Hamiltonian $\widetilde{H}_{AB}^{+}$. The state $\rho_t$ of the system at point of time $t$ is given by 
\begin{align}
   \psi_t & = \frac{1}{2} \left(1+(2 p-1) \cos \left(2 \theta t\right)\right)\ket{00}\bra{00} \nonumber\\
   & \hspace{0.35cm}+ (\sqrt{p(1-p) }+\frac{\iota}{2}  (2 p-1) \sin \left(2 \theta t\right)) \ket{00}\bra{11} \nonumber\\
   & \hspace{0.35cm}+ (\sqrt{p(1-p)} + \frac{\iota}{2}  (1-2 p) \sin \left(2 \theta t\right)) \ket{11}\bra{00} \nonumber\\
   & \hspace{0.35cm}+ \frac{1}{2} \left(1+(1-2 p) \cos \left(2 \theta t\right)\right) \ket{11}\bra{11}, \label{equ:final_state_in_non-local_dynamics}
\end{align}
where $\theta:=\mu_{x} - \mu_{y}$. To estimate bounds on the negativity~\eqref{eq:theorem-1} and the concurrence~\eqref{eq:theorem-2}, we need the following quantities:

\begin{align}
    {\mathscr{C}}^2(\psi_{0}) &= 2 \left(\left|p (p-1) \right| -p(p-1)\right),\\
    {{\tr(\psi_{0}(\widetilde{H}_{AB}^{+})^2)}} &= {\theta^2+\mu _z^2},\\ 
      \norm{\cal{L}_t(\psi_{AB}^{T_{B}})}_{1} &= f(\theta,p)(| \sin(2 t \theta)|+ | \cos (2 t \theta )|),\\
     \mathscr{N}(\psi_{0}) &= \sqrt{p\left(1-p\right)},
\end{align}
 \begin{align}
     \mathscr{N}(\psi_t) &=\frac{\sqrt{-4 p^2-(1-2 p)^2 \cos (4 \theta  t)+4 p+1}}{2 \sqrt{2}}\label{equ:negativity_for_non-local_dynamics},\\
      {\mathscr{C}}^2(\psi_{t}) &= \frac{1}{2} \left(4 \left|p\left(p-1\right) \right| -(1-2 p)^2 \cos \left(4 \theta t\right)+1\right) \label{equ:concurrence_for_non-local_dynamics},
\end{align}

where $f(\theta,p) := 2\theta\left| 1-2 p\right|$.

\begin{figure}[ht]
    \centering
    \includegraphics[width=8cm]{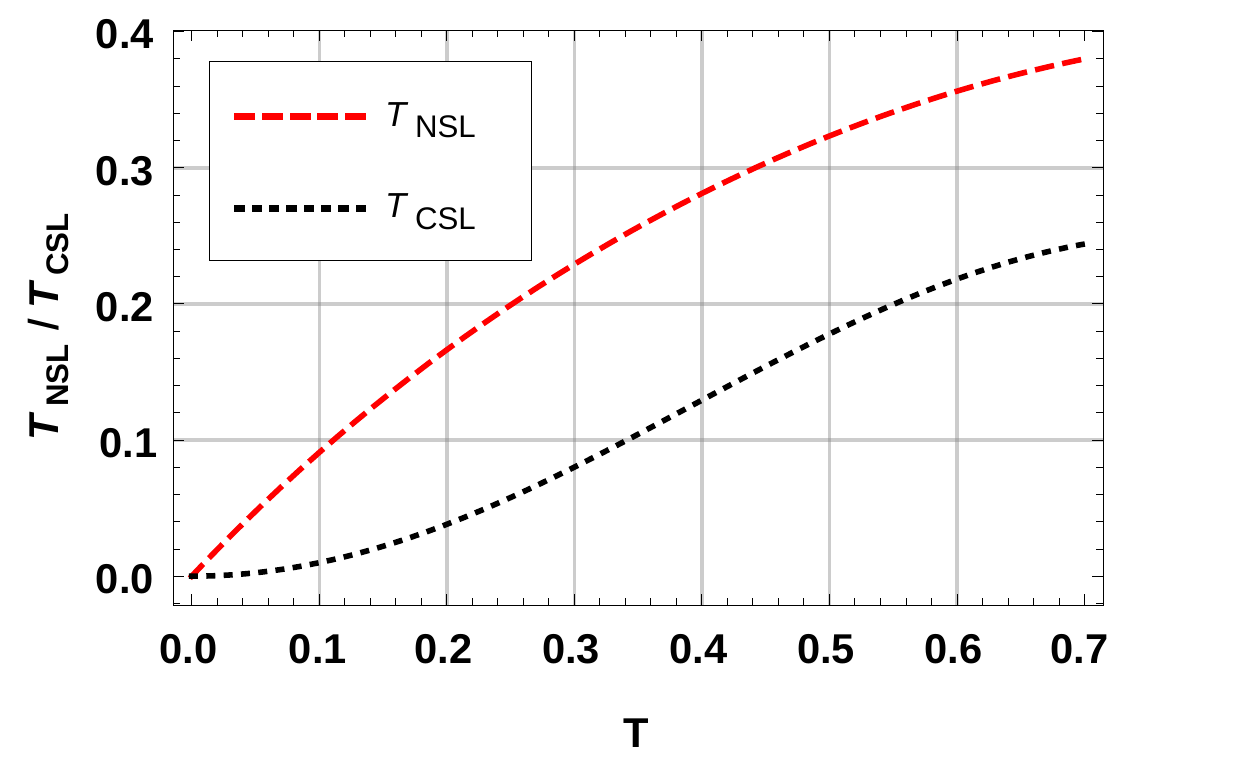}
    \caption{Here we depict $T_{\rm CSL}$ and $T_{\rm NSL}$ vs $T$ $\in$ $[0,0.7]$ for given unitary process and we have taken $\theta =1$, $p = 0$ and $\mu_{z} = 0.1$. If we plot $T_{\rm CSL}$ and $T_{\rm NSL}$ vs $T$ beyond the above mentioned range the speed limit time start decreasing due to decrease in value of negativity and concurrence. }
    \label{fig:speed_limit_for_non-local_hamiltonians}
\end{figure}

In Fig.~\ref{fig:speed_limit_for_non-local_hamiltonians}, we plot $T_{\rm NSL}$~\eqref{eq:theorem-1} and $T_{\rm CSL}$~\eqref{eq:theorem-2} vs $T$ $\in$ $[0,0.7]$ for unitary dynamics generated by two-qubit non-local Hamiltonian $\widetilde{H}^{+}_{AB}$ with $\theta=\mu_x-\mu_y=1$, $\mu_{z} = 0.1$, and initial state of the system to be $\psi_0$ with $p=0$. In Fig.~\ref{fig:speed_limit_for_non-local_hamiltonians}, we observe that under the given unitary process, the concurrence of the given system evolves faster than its negativity. It is clear from Eq.~\eqref{equ:final_state_in_non-local_dynamics} that density operator is function of $\theta$ (i.e. $\mu_{x}-\mu_{y}$) and so do the entanglement measures such as the negativity (Eq.~\eqref{equ:negativity_for_non-local_dynamics}) and square of the concurrence (Eq.~\eqref{equ:concurrence_for_non-local_dynamics}). The nature of both entanglement monotones is periodic in time and the period depends on $\theta$ (see Appendix~\ref{Negativity_and_Concurrence_For_Non-Local_Hamiltonian}). We found that the bounds~\eqref{eq:theorem-1} and~\eqref{eq:theorem-2} are relatively tighter for small values of $\theta$ (i.e., $\theta\in (0,1]$) in comparison to larger values of $\theta$ (i.e., $\theta>1$). See Fig.~\ref{fig:speed_limit_for_non-local_hamiltonians} and Appendix~\ref{Negativity_and_Concurrence_For_Non-Local_Hamiltonian}. Here, by relatively tighter we mean the gap between the evolution time and the time obtained from the speed limits.
 
{\it Open quantum dynamics.---} Before we analytically and numerically compute speed limits on the negativity~\eqref{eq:theorem-1} and Bell-CHSH observable~\eqref{eq:theorem-3} for bipartite quantum systems evolving under non-unitary quantum dynamics, let us briefly recall some concepts from open quantum systems. In case of open quantum systems we assume that dynamics of extended system (system + environment) is unitary and after tracing out the environment, we get the evolution equation for open quantum system. Under Markovian approximation, Eq.~\eqref{Master_equation_density_operator} and Eq.~\eqref{Master_equation_density_observable} reduce to the following~\cite{Lindblad1976,Gorini1975}:
 \begin{align}
       \dot{\rho_t}=-\iota[H,\rho_{t}]+\sum_{\alpha}\left(2 L_{\alpha}\rho_{t} L^{\dag}_{\alpha}-\{L^{\dag}_{\alpha}L_{\alpha},\rho_{t}\}\right),\\
        \dot{\cal{O}_{t}}=\iota[H,\cal{O}_{t}]+\sum_{\alpha}\left(2 L^{\dag}_{\alpha}\cal{O}_{t} L_{\alpha}-\{L^{\dag}_{\alpha}L_{\alpha},\cal{O}_{t}\}\right),
 \end{align}
where $\{O_1,O_2\}:= O_1O_2 + O_2O_1$ denotes anti-commutator bracket, $H$ is the Hamiltonian of the system, and $L_{\alpha}$'s (represent coupling between system and environment) are called Lindbladian operators or quantum jump operators. The above equations are called Lindblad-Gorini-Kossakowski-Sudarshan (LGKS) master equations.    

Here we consider two spin-$1/2$ particles $A$ and $B$ each coupled with environments $E_{A}$ and $E_{B}$, respectively, where $E_A$ and $E_B$ are not interacting with each other. We assume that the system $AB$ is initialised in a bipartite pure state  $\rho_0$ of the form
\begin{align}
     \rho_0 &=p \ket{00}\bra{00} + \sqrt{\left(1-p\right) p}(\ket{00}\bra{11}+\ket{11}\bra{00}) \nonumber \\
      & \hspace{0.35cm}+ \left(1-p\right) \ket{11}\bra{11},
\end{align}
 where $p\in[0,1]$.
 We consider Bell-CHSH observable $\cal{B}_{0} = \hat{a}.\vec{\sigma} \otimes \left(\hat{b}+\hat{b'}\right).\vec{\sigma} + \hat{a'}.\vec{\sigma} \otimes \left(\hat{b}-\hat{b'}\right).\vec{\sigma}$ with initial settings (at $t=0$) being $\hat{b} = \cos(\eta)\hat{z} + \sin(\eta)\hat{z}$, $\hat{a} = \hat{z} $, $\hat{a'} = \hat{x}$, $\hat{b'} = \cos(\eta)\hat{z} - \sin(\eta)\hat{x}$ and $\tan(\eta) = 2\sqrt{p\left(1-p\right)} $~\cite{Popescu1992}. We note that the settings for Bell-CHSH test need not remain optimal as the settings evolve during the dynamical process. 
 
{\it Pure dephasing process.---} We first consider pure dephasing channel as an example of quantum correlation degradation process. The Lindbladian operators for pure dephasing process are given as $L_{1} = \sqrt{\frac{\gamma^A}{2}} \sigma^A_{z}\otimes \mathbbm{1}_{B}$ and $L_{2} = \sqrt{\frac{\gamma^B}{2}}\mathbbm{1}_{A} \otimes\sigma^B_{z} $, where $\sigma^A_{z}$ and $\sigma^B_{z}$ are Pauli operators acting on system $A$ and $B$, respectively, and $\gamma^A,\gamma^B\in\mathbbm{R}$ denote the strength of dephasing. The LGKS master equation governs time-evolution of the state $\rho_{t}$ in Schrodinger picture and Bell-CHSH observable $\cal{B}_{t}$ in Heisenberg picture,
 \begin{align}\label{rho:dep}
  \frac{\rm d}{{\rm d}t}\rho_{t} &= \gamma^A\left(\sigma^A_{z}\otimes \mathbbm{1}_{B}\left(\rho_{t}\right)\sigma^A_{z}\otimes \mathbbm{1}_{B} -\rho^{AB}_{t}\right)\nonumber\\
  & \hspace{0.35cm}+\gamma^B\left( \mathbbm{1}_{A}\otimes\sigma^B_{z}\left(\rho^{AB}_{t}\right)\mathbbm{1}_{A}\otimes\sigma^B_{z} -\rho_{t}\right),\\
\label{bell:dep}
  \frac{\rm d}{{\rm d}t}\cal{B}_{t} &= \gamma^A\left(\sigma^A_{z}\otimes \mathbbm{1}_{B}\left(\cal{B}_{t}\right)\sigma^A_{z}\otimes \mathbbm{1}_{B} -\cal{B}_{t}\right)\nonumber\\
  & \hspace{0.35cm}+\gamma^B\left( \mathbbm{1}_{A}\otimes\sigma^B_{z}\left(\cal{B}_{t}\right)\mathbbm{1}_{A}\otimes\sigma^B_{z} -\cal{B}_{t}\right).
  \end{align}
The respective solutions to Eq.~\eqref{rho:dep} and Eq.~\eqref{bell:dep} are
  \begin{align}
      \rho_{t} &= p \ket{00}\bra{00} +\sqrt{p\left(1-p\right)} {\rm e}^{-4 \gamma  t}\left( \ket{00}\bra{11}+\ket{11}\bra{00}\right) \nonumber \\
      & \hspace{0.35cm}+ \left(1-p\right) \ket{11}\bra{11},\\
     \cal{B}_{t} & = 2 \cos (\eta )\left(\ket{00}\bra{00}-\ket{01}\bra{01}-\ket{10}\bra{10}+\ket{11}\bra{11}\right) \nonumber \\
     & \hspace{0.35cm} + 2 \sin (\eta ) {\rm e}^{-4 \gamma  t}\left(\ket{00}\bra{11}+\ket{10}\bra{01}+\ket{01}\bra{10}\right. \nonumber\\
     & \hspace{0.35cm}+\left.\ket{11}\bra{00}\right),
  \end{align}
  where we have assumed that dephasing rate of both environments are equal to $\gamma$. To estimate bounds on the negativity ~\eqref{eq:theorem-1} and Bell-CHSH observable ~\eqref{eq:theorem-3}, we need the following quantities:
\begin{align}
 \left| \mathscr{N}(\rho_{T})-\mathscr{N}(\rho_{0})\right| &= \sqrt{p - p^2}\left(1-{\rm e}^{-4 \gamma t}\right), \\
  \norm{\cal{L}_t(\rho_{t}^{T_{B}})}_{1} &= 8 \gamma \left(\sqrt{p-p^2}\right) {\rm e}^{-4 \gamma t},\\
 \left| \langle\cal{B}_{T}\rangle_{\rho_{0}}-\langle\cal{B}_{0}\rangle_{\rho_{0}}\right| &= 4 \sqrt{\left(p-p^2\right)}\left| \left(1-{\rm e}^{-4 \gamma t}\right)  \sin (\eta )\right|,\\
  \min\left\{\Lambda_{T}^{\infty},\Lambda_{T}^{1},\Lambda_{T}^{2}\right\}& = 8 \gamma  \sin (\eta ) {\rm e}^{-4 \gamma  t}.
  \end{align}
We can analytically verify that the speed limit~\eqref{eq:theorem-1} on the negativity is tight (we get $T=T_{\rm NSL}$) for arbitrary choice of parameter $\gamma$ and hence it is attainable for pure dephasing process.
  
\begin{figure}[h!]
    \centering
     \includegraphics[width=8cm]{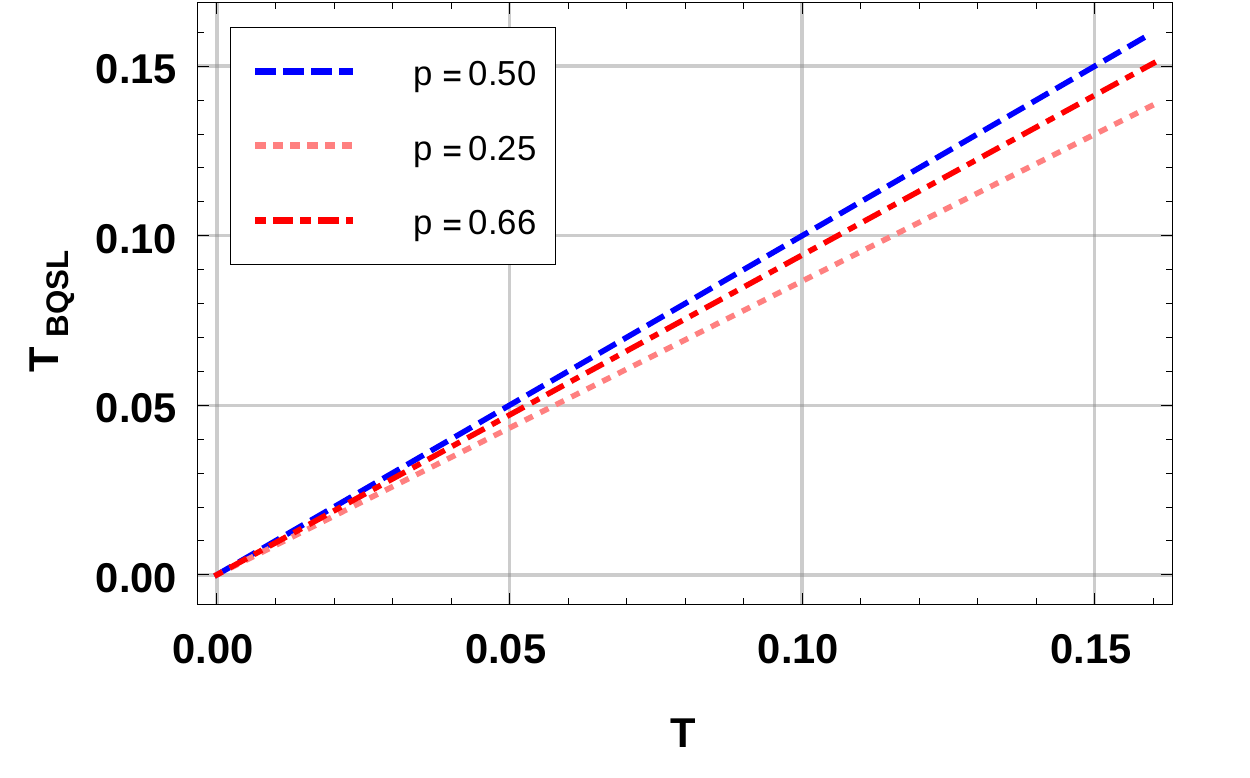}
    \caption{Here we depict $T_{\rm BQSL}$~\eqref{eq:theorem-3} vs $T$ bounds for pure dephasing process and we have taken $\gamma = 1$ and $p\in\{0.25,0.50,0.66\}$.}
  \label{fig:Pure_Dephasing}
\end{figure}

In Fig.~\ref{fig:Pure_Dephasing}, we plot $T_{\rm BQSL}$~\eqref{eq:theorem-3} vs $T$ $\in$ $[0,0.15]$ for pure dephasing process and we have considered $\gamma = 1$ and $p\in\{0.25,0.50,0.66\}$. We find that the the Bell-nonlocal correlation degrade faster for $p\in\{0.25,0.66\}$ (non-maximally entangled state) in comparison to $p=0.50$ (maximally entangled state). Furthermore, we found that the bound~\eqref{eq:theorem-3} tight and attainable for pure dephasing process when $\gamma = 1$ and $p=0.50$.

 {\it Depolarising process}.--- We next consider depolarising process as an example of quantum correlation degradation process. For depolarising process, the Lindbladian operators are given as $L^A_{i} = \sqrt{\frac{\gamma^A}{8}} \sigma^A_{i}\otimes \mathbbm{1}_{B}$ and $L^B_{j} = \sqrt{\frac{\gamma^B}{8}}\mathbbm{1}_{A} \otimes\sigma^B_{j} $ for $i,j\in\{1,2,3\} $, where $\sigma^A_{i}$ and $\sigma^B_{j}$ are Pauli operators acting on $A$ and $B$, respectively, and $\gamma^A,\gamma^B\in\mathbbm{R}$ denote the strength of depolarising. The time evolution of bipartite state  $\rho_{t}$ in Schr\"odinger's picture and Bell-CHSH observable $\cal{B}_{t}$ in Heisenberg's picture respectively are governed by LGKS master equation and given as 
 \begin{align}\label{rho:depo}
      \frac{\rm d}{{\rm d}t}\rho_{t} &= \frac{\gamma^A}{4} \sum_{i =1}^{3}\left(\sigma^A_{i}\otimes \mathbbm{1}_{B}\left(\rho_{t}\right)\sigma^A_{i}\otimes \mathbbm{1}_{B} -\rho_{t}\right)\nonumber\\
  & \hspace{0.35cm}+\frac{\gamma^B}{4}\sum_{i=1}^{3}\left( \mathbbm{1}_{A}\otimes\sigma^B_{i}\left(\rho_{t}\right)\mathbbm{1}_{A}\otimes\sigma^B_{i} -\rho_{t}\right),\\  \label{bell:depo}
      \frac{\rm d}{{\rm d}t}\cal{B}_{t} &= \frac{\gamma^A}{4} \sum_{i =1}^{3}\left(\sigma^A_{i}\otimes \mathbbm{1}_{B}\left(\cal{B}_{t}\right)\sigma^A_{i}\otimes \mathbbm{1}_{B} -\cal{B}_{t}\right)\nonumber\\
  & \hspace{0.35cm}+\frac{\gamma^B}{4}\sum_{i=1}^{3}\left( \mathbbm{1}_{A}\otimes\sigma^B_{i}\left(\cal{B}_{t}\right)\mathbbm{1}_{A}\otimes\sigma^B_{i}- \cal{B}_{t}\right).
 \end{align}
 The respective solutions to Eq.~\eqref{rho:depo} and Eq.~\eqref{bell:depo} are
  \begin{align}
      \rho_{t} &= \frac{1}{2} {\rm e}^{-\gamma t}\left(2 p+\cosh (\gamma  t)-1\right)\ket{00}\bra{00}\nonumber \\
      &\hspace{0.35cm}+ \sqrt{p\left(1-p\right)} {\rm e}^{-2 \gamma  t} \left(\ket{00}\bra{11}+\ket{11}\bra{00}\right)\nonumber \\
      &\hspace{0.35cm}+\frac{1}{2} {\rm e}^{-\gamma t} \sinh (\gamma  t)\left(\ket{01}\bra{01}+\ket{10}\bra{10}\right)\nonumber \\
      & \hspace{0.35cm}+ \frac{1}{2} {\rm e}^{-\gamma t} \left(1-2 p+\cosh (\gamma  t)\right) \ket{11}\bra{11},\\
      \cal{B}_{t} & = q\left(\ket{00}\bra{00}+\ket{11}\bra{11}\right)+h\left(\ket{01}\bra{01}+\ket{10}\bra{10}\right) \nonumber \\
     & \hspace{0.35cm} + 2 {\rm e}^{-2 \gamma  t}\left(\ket{00}\bra{11}+\ket{10}\bra{01}+\ket{01}\bra{10}\right. \nonumber\\
     & \hspace{0.35cm}+\left.\ket{11}\bra{00}\right),
  \end{align}
  where we assumed $\gamma^{A}=\gamma^{B}=\gamma$, $q=\frac{1}{2} {\rm e}^{-\gamma t} \left(4 \cos (\eta ) \cosh (\gamma  t)-4 \cos (\eta ) \sinh (\gamma  t)\right)$, and $h = \frac{1}{2} {\rm e}^{-\gamma t} (4 \cos (\eta ) \sinh (\gamma  t)-2 \cos (\eta ) (\cosh (\gamma  t)-1)-2 \cos (\eta ) (\cosh (\gamma  t)+1))$. To estimate bounds on the negativity~\eqref{eq:theorem-1} and Bell-CHSH observable~\eqref{eq:theorem-3}, we need the following quantities:
\begin{align}
  \left| \mathscr{N}(\rho_{T})-\mathscr{N}(\rho_{0})\right| &= \frac{1}{4} \left(1+4 \sqrt{(1-p) p}\right) {\rm e}^{-2 \gamma  t} \left({\rm e}^{2 \gamma  t}-1\right) \label{equ:negativity_difference_for_depolarising_process}, \\
 \norm{\cal{L}_t(\rho_t^{T_{B}})}_{1} &= \frac{{\rm e}^{-\frac{\gamma t}{2}}}{2}\left( {\rm e}^{\frac{\gamma t}{2}}\left(\sqrt{a'-b'}+ \sqrt{a'+b'}\hspace{0.1cm}\right)\right.\nonumber \\
 & \hspace{0.35cm}+\left.\sqrt{\gamma}\left(\sqrt{a}+\sqrt{b}\right) \right),\\
   | \langle\cal{B}_{T}\rangle_{\rho_{0}}-\langle\cal{B}_{0}\rangle_{\rho_{0}}| &= 2 \left(1-{\rm e}^{-2 \gamma  t}\right)\left|g(\eta,p) \right|,\\
   \min\left\{\Lambda_{T}^{\infty},\Lambda_{T}^{1},\Lambda_{T}^{2}\right\}& = 2 \gamma  {\rm e}^{-2 \gamma  t} \bigg(\sqrt{2-2 \sqrt{\cos ^2(2 \eta )}}\nonumber \\
  &\hspace{0.35cm}+\sqrt{1-\sin (2 \eta )}+\sqrt{\sin (2 \eta )+1}\bigg),
\end{align}
where 
$g(\eta,p)= \cos (\eta )+2 \sin (\eta ) \sqrt{p\left(1-p\right)}$, $a' = \gamma ^2 {\rm e}^{-4 \gamma  t}\left (1-16p \left(p-1\right) \right)$ , $ a = \gamma {\rm e}^{-3 \gamma  t} \left(\left(1-2 p\right) {\rm e}^{\gamma  t}-1\right)^2 $, $b' =8\sqrt{ \gamma ^4{\rm e}^{8 \gamma  t}p\left (1-p\right)} $, and $b = 8\sqrt{\gamma {\rm e}^{-3 \gamma  t} \left(1+(1-2 p) {\rm e}^{\gamma  t}\right)^2}$. The above expression for $ \min\left\{\Lambda_{T}^{\infty},\Lambda_{T}^{1},\Lambda_{T}^{2}\right\}$ is only valid in the range of $p \in[0.05,0.95]$ and Eq.~\eqref{equ:negativity_difference_for_depolarising_process} is only valid in range $0\leq \gamma \leq 1$.
\begin{figure}
     \centering
     \begin{subfigure}[b]{0.41\textwidth}
         \centering
         \includegraphics[width=\textwidth]{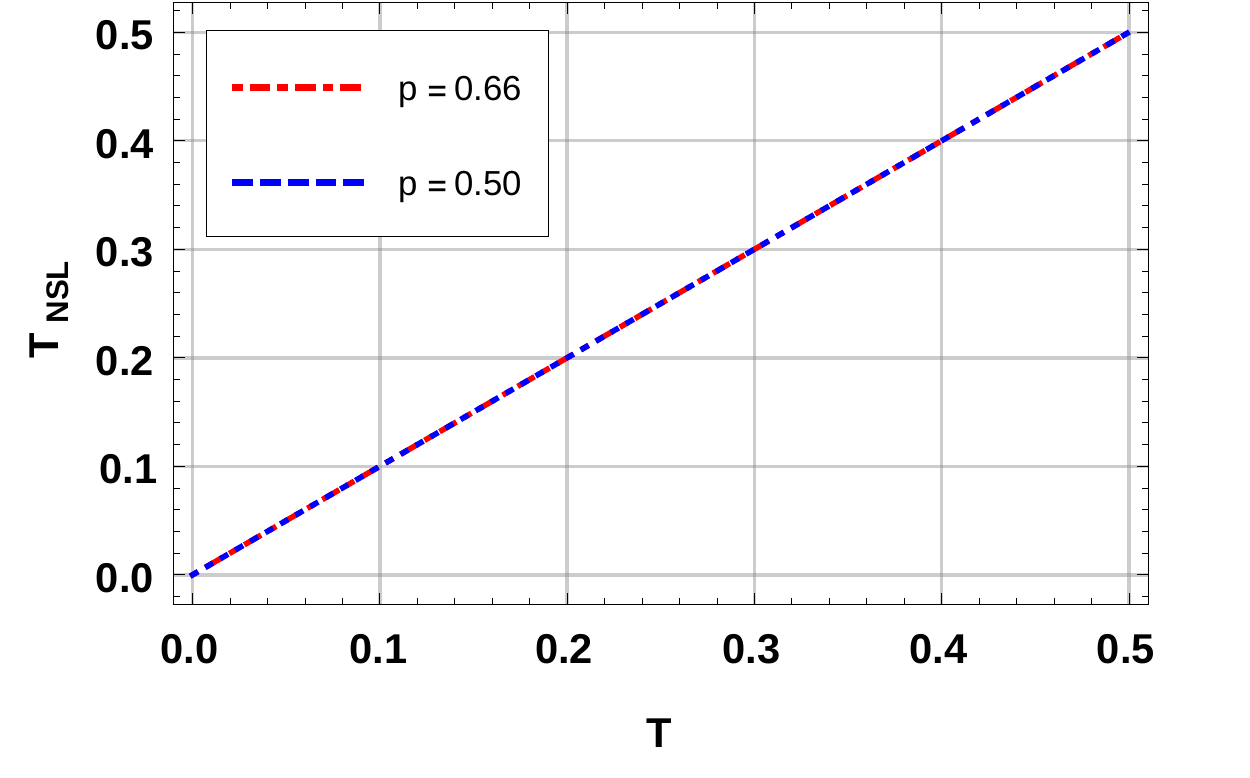}
         \caption{$T_{\rm NSL}$ vs $T$ with initial state parameter $p\in\{0.50,0.66\}$.}
         \label{fig:Negativity_for_depolarising_process}
     \end{subfigure}
     \hfill
     \begin{subfigure}[b]{0.41\textwidth}
         \centering
         \includegraphics[width=\textwidth]{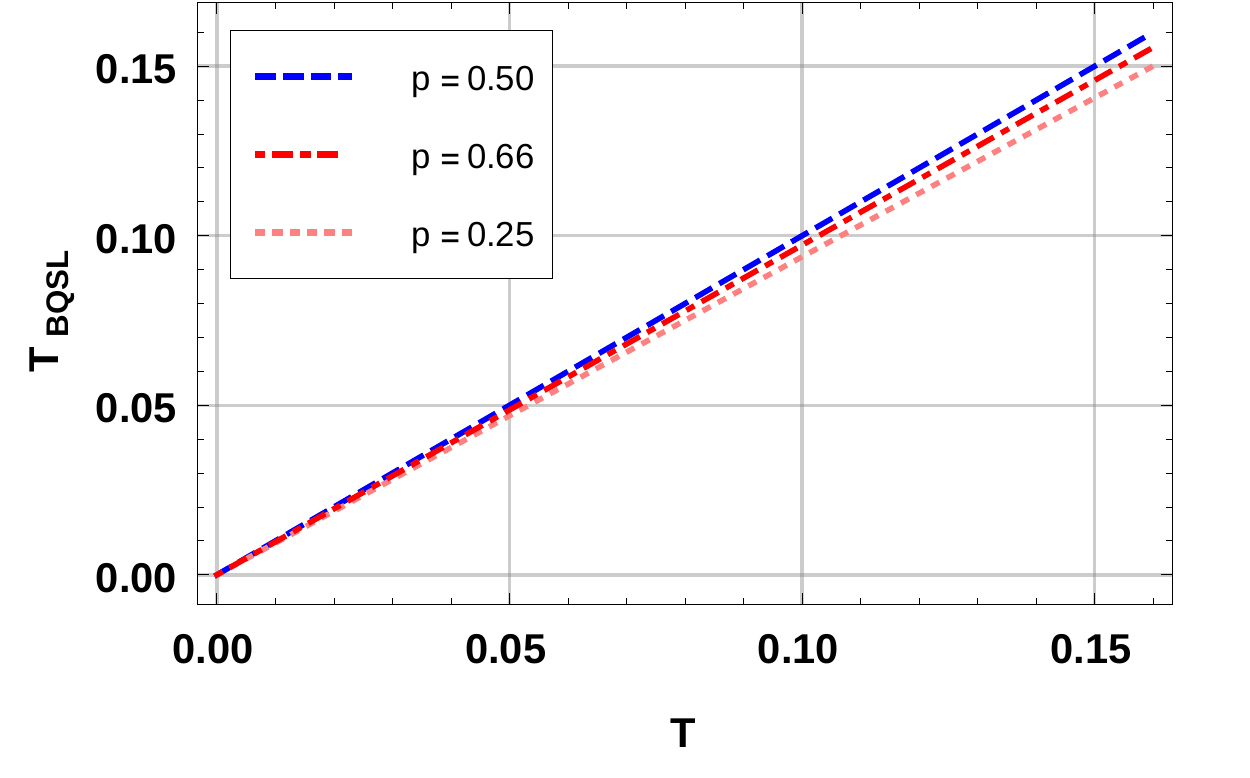}
         \caption{$T_{\rm BQSL}$ vs $T$ with initial state parameter $p\in\{0.25,0.50,0.66\}$.}
         \label{fig:Bell_Operator_Speed_Limit_for_depolarising_process}
     \end{subfigure}
 \caption{For depolarising process with $\gamma = 1$, we depict $T_{\rm NSL}$~\eqref{eq:theorem-1} vs $T$ in (a) and $T_{\rm BQSL}$~\eqref{eq:theorem-3} vs $T$ in (b).}
 
    \label{fig:Depolarisation}
\end{figure}

In Fig.~\ref{fig:Depolarisation}, we plot $T_{\rm NSL}$~\eqref{eq:theorem-1} vs $T$ and $T_{\rm BQSL}$~\eqref{eq:theorem-3} vs $T$ for depolarising process and we have considered $\gamma = 1$. For $T_{\rm BQSL}$ we have taken $p\in\{0.25,0.50,0.66\}$ and for $T_{\rm NSL}$ we have taken $p\in\{0.50,0.66\}$. We observe that the Bell-nonlocal correlation degrades faster for $p\in\{0.25,0.66\}$ (non-maximally entangled state) in comparison to $p=0.50$ (maximally entangled state). We note that the bound~\eqref{eq:theorem-3} is tight and attainable for the given depolarising process when $\gamma = 1$ and $p=0.50$. Also, the bound~\eqref{eq:theorem-1} is tight and attainable for the given depolarising process when $\gamma=1$ and $p=\{0.50,0.66\}$.
We also have estimated the speed limits on entanglement and Bell-CHSH observable. See Appendix \ref{fig:Amplitude-Damping} for detailed calculation.

\section{Conclusion}

 An important aspect of developing quantum devices is linked to the understanding of the rates at which quantum resources are consumed or generated~\cite{Aifer2022,S.Das2018,S_Das2021}. The rate at which correlations change for any given quantum process is critical for designing the architecture of quantum processors as well as to control the dynamical evolution of quantum systems for desired information processing or computation tasks~\cite{Ladd2010, Dowling2003,Jones2012}. In this work, we introduced speed limits on some of the entanglement monotones like negativity and concurrence and also some other correlations like quantum mutual information and Bell-CHSH observable. These speed limits provide lower bound on the minimal time required to bring certain amount of change in these correlations for bipartite quantum systems undergoing time-evolution. We have also discussed a few cases of practical processes for which some of these speed limits are attainable and hence can be considered to be tight. As a byproduct of the scenarios we consider in this paper, we are also able to improve upon the speed limit on the von Neumann entropy derived in Ref.~\cite{Mohan2022}. It is important to note that standard speed limits answer the fundamental question of how fast quantum systems can evolve over time, while our work provides limits on the rates at which we can create or destroy correlations through physical processes or how fast quantum devices can consume these correlations to perform a given task (as correlations can be valuable resources in quantum information and computing tasks. As a future direction, we leave open questions to derive (tight) speed limits on correlational measures beyond those discussed in this paper, for an instance, multipartite nonlocal quantum correlations~\cite{Bancal2013,Home2015,Horodecki2021}.

\begin{acknowledgements}
VP, DS, and BM thank Sohail, Chirag Srivastava, Abhay Srivastav, and Ujjwal Sen for useful discussions. BM acknowledges the support of the INFOSYS scholarship. AKP acknowledges support of the J. C. Bose Fellowship from the Department of Science and Technology (DST) India under Grant No.~JCB/2018/000038 (2019–2024).
\end{acknowledgements}

\appendix

\section{Partial Transpose operation}\label{pattial_transpose}
For an arbitrary operator $X_{AB}$ defined on $\cal{H}_{AB}$  and a fixed orthonormal basis $\{\ket{\alpha}\}_{A}$, the partial transpose $\Gamma^{A}$  is defined as the following linear map
\begin{equation}
    X_{AB}^{\Gamma_A}:= \Gamma^{A}(X_{AB})= \sum_{\alpha,\beta}(\ket{\alpha}\bra{\beta}\otimes\mathbbm{1}_{B})X_{AB}(\ket{\alpha}\bra{\beta}\otimes\mathbbm{1}_{B}),
\end{equation}
where $\Gamma^{A}$ is tensor product of two maps $T^{A}$ (transpose operator on $\cal{H}_{A}$) and $\mathbbm{1}_{B}$. Map $\Gamma^{B}$ can also be defined in similar way. 

A bipartite quantum state $\rho_{AB}$ is called separable if it can be written in the following form:
\begin{equation}
    \rho_{AB} = \sum_{i}p_{i}\op{\psi_{i}}_{A}\otimes\op{\phi_{i}}_{B},
\end{equation}
for some probability distribution $\{p_{i}\}_i$ and sets $\{\ket{\psi_{i}}_{A}\}_i$ and $\{\ket{\phi_{i}}_{B}\}_i$ of pure states. Let $\operatorname{SEP}(A;B)$ denotes the set of separable states defined on $\mathcal{H}_{AB}$. Under the action of partial transpose map, set of separable states remains a closed set i.e., $\Gamma^{A}(\rho_{AB})\in \operatorname{SEP}(A;B), \forall \rho_{AB}\in \operatorname{SEP}(A;B)$. If the action of partial transposition map on a state yields operator which is positive semidefinite, then the state is said to be positive partial transposition (PPT)~(see ). That is, $\rho_{AB}\in\cal{D}(\cal{H}_{AB})$ is a PPT state if $\Gamma^{B}(\rho_{AB})\geq0$, which also implies that $\Gamma^{A}(\rho_{AB})\geq 0$. If a quantum state $\rho_{AB}\in\cal{D}(\cal{H}_{AB})$ corresponds to a separable bipartite state then $\Gamma^{B}(\rho_{AB})\geq0$ but the converse statement is true if and only if both the quantum systems $A$ and $B$ are either qubit or one is qubit and the other is qutrit~\cite{Horodecki1996,Peres1996}. {For a $1\oplus 1$-mode Gaussian state $\rho_{AB}$, the positive partial transposition is necessary and sufficient criterion for  separability~\cite{Simon2000}}. If $\rho_{AB}$ is not separable then it is entangled. There are some special type of entangled states called ``maximally entangled" states. If $\rho_{AB}$ is maximally entangled and we perform any local measurement on subsystem $A$ or $B$, then we gain no information about the preparation of the state; instead we merely generate a random bit.

\section{Operator Functions} \label{operator_function}
Let $\cal{O}$ be a normal operator acting on a Hilbert space $\cal{H}$, i.e., $\cal{O}^\dag\cal{O}=\cal{O}\cal{O}^{\dag}$. The kernel of $\cal{O}$
is the span of the eigenvectors of  $\cal{O}$ corresponding to its zero eigenvalues. The  support $\operatorname{supp(\cal{O})}$ of $\cal{O}$ is the subspace of $\cal{H}$ orthogonal to its kernel. Let $ \{\alpha_{i}\}_i$ be the set of eigenvalues of $\cal{O}$ and $\{\ket{\alpha_{i}}\}_i$ be the corresponding set of eigenvectors of $\cal{O}$. Then $\cal{O}$ can be written as follows:
\begin{equation}
    \cal{O} =\sum_{i} \alpha_{i} \op{\alpha_{i}},
\end{equation}
which is called a spectral decomposition of $\cal{O}$. The projection onto the support $\operatorname{supp}(\cal{O})$ of $\cal{O}$ is denoted by
\begin{equation}
    \Pi_{\cal{O}} = \sum_{i:\alpha_{i}\neq 0}\op{\alpha_{i}}.
\end{equation}
If $f$ is a real valued function with domain ${\rm Dom}(f)$, then $f(\cal{O})$ is defined as 
\begin{equation}
    f(\cal{O}) = \sum_{i:\alpha_{i}\in {\rm Dom}(f)} f(\alpha_{i}) \op{\alpha_{i}}.
\end{equation}
\section{ Speed limit for observable}\label{QSL:Observable}
 The expectation value of the observable $\cal{O}_t$ is given as 
\begin{equation}
 \langle\cal{O}_{t}\rangle_{\rho} ={\rm tr}(\cal{O}_{0}\Phi_{t}(\rho)) = {\rm tr}(\Phi^{\dagger}_{t}(\cal{O}_{0})\rho),
\end{equation}
where $\Phi_{t}$ is generator of dynamics, $\cal{O}_{t}= \Phi_{t}(\cal{O}_{0})$ and $\rho$ is the given state of bipartite  quantum system. The time evolution of the observable $\cal{O}$ is given by the following equation,
\begin{equation}\label{Lind}
 \frac{{\rm d} }{{\rm d}t}{ \cal{O}_{t}} =   {\mathcal{L}^{\dagger}}[\cal{O}_{t}].
\end{equation}
Let us take the average of Eq~\eqref{Lind} in the bipartite state $\rho$ and its absolute value. By applying the  H\"older's inequality inequality, we obtain the following inequality 
\begin{align}\label{rate:4}
\left|\frac{{\rm d} }{{\rm d}t}{ \langle \cal{O}_{t} \rangle_{\rho}}\right| \leq \norm{\rho}_{1}\norm{{\mathcal{L}^{\dagger}}[\cal{O}_{t}]}_{\rm \infty} \leq \norm{\rho}_{1}\norm{{\mathcal{L}^{\dagger}}[\cal{O}_{t}]}_{\rm 2}\nonumber \\ \leq \norm{\rho}_{1}\norm{{\mathcal{L}^{\dagger}}[\cal{O}_{t}]}_{\rm 1}.
\end{align}
The above inequality~\eqref{rate:4} is the upper bound on that the rate of change of expectation value of the observable evolving under given dynamics. After integrating the above equation with respect to time $t$, we then obtain the following unified bound
\begin{equation}\label{Bound:open1}
 T \geq  T_{\rm OQSL}= \frac{| \langle\cal{O}_{T}\rangle_{\rho}-\langle\cal{O}_{0}\rangle_{\rho}|}{\norm{\rho}_{1}}\max\left\{\frac{1}{\Lambda_{T}^{\infty}},\frac{1}{\Lambda_{T}^{1}},\frac{1}{\Lambda_{T}^{2}}\right\},
\end{equation}
where $\Lambda_{T}^{\alpha}=\frac{1}{T}\int_{0}^{T}{\rm d}t\norm{\cal{L}_{t}^{\dagger}(\cal{O}_{t})}_{\alpha}$ for $\alpha\in\{1,2,\infty\}$ is the evolution speed of the observable of the given system. For the given dynamics, the bound given in \eqref{Bound:open1} determines how fast the expectation value of the observable changes in time (cf.~\cite{B.Mohan2021}).

\section{Hermiticity after partial transposition}\label{Herm}
Consider a bipartite system $AB$ with Hilbert space $\cal{H}_{AB}$ and density matrix $\rho_{AB}$. If $\{\ket{\alpha_{i}}_{A}\}_{i}$ are the orthonormal basis of $\cal{H}_{A}$ and $\{\ket{\beta_{i}}_{B}\}_{i}$ are the orthonormal basis of $\cal{H}_{B}$  then the density matrix of $AB$ takes the general form

\begin{align}
     \rho_{AB}&=\sum_{\{\alpha,\alpha',\beta,\beta'\}}a^{\alpha,\alpha'}_{\beta,\beta'}\ket{\alpha}\bra{\alpha'}_A\otimes\ket{\beta}\bra{\beta'}_B,\\
     \rho^\dag_{AB}&=\sum_{\{\alpha,\alpha',\beta,\beta'\}}\left(a^{\alpha,\alpha'}_{\beta,\beta'}\right)^*\ket{\alpha'}\bra{\alpha}_{A}\otimes\ket{\beta'}\bra{\beta}_{B}\nonumber\\
     &=\sum_{\{\alpha,\alpha',\beta,\beta'\}}\left(a^{\alpha',\alpha}_{\beta',\beta}\right)^*\ket{\alpha}\bra{\alpha'}_A\otimes\ket{\beta}\bra{\beta'}_B.
\end{align}
 The density operator $\rho_{AB}$ is a hermitian i.e. $\rho^\dag_{AB}=\rho_{AB}$. Now, comparing coefficients of $\rho_{AB}$ and  $\rho^\dag_{AB}$, we  find the following condition on the coefficient $a^{\alpha,\alpha'}_{\beta,\beta'}$
 \begin{equation}
     a^{\alpha,\alpha'}_{\beta,\beta'}=\left(a^{\alpha',\alpha}_{\beta',\beta}\right)^*\label{equ:coefficient_comparision}.
 \end{equation}
The partial transpose $\rho^{AB}$ can be written as
\begin{align}
    \rho^{T_{B}}_{AB}&=\sum_{\{\alpha,\alpha',\beta,\beta'\}}a^{\alpha,\alpha'}_{\beta,\beta'}\ket{\alpha}\bra{\alpha'}_A\otimes\ket{\beta'}\bra{\beta}_B, \\
(\rho^{T_{B}}_{AB})^{\dag}&=\sum_{\{\alpha,\alpha',\beta,\beta'\}}\left(a^{\alpha,\alpha'}_{\beta,\beta'}\right)^*\ket{\alpha'}\bra{\alpha}_{A}\otimes\ket{\beta}\bra{\beta'}_{B}\nonumber\\
&=\sum_{\{\alpha,\alpha',\beta,\beta'\}}\left(a^{\alpha',\alpha}_{\beta',\beta}\right)^*\ket{\alpha}\bra{\alpha'}_{A}\otimes\ket{\beta'}\bra{\beta}_{B}.
\end{align}

Using~\eqref{equ:coefficient_comparision} it can be shown that partial transpose of a density operator is hermitian.

\section{Alternative proof of Theorem \ref{thm:2}}\label{alterproof}
\begin{proof}
 The negativity of time evolved bipartite state $\rho_{t}$ given by
\begin{equation}
\mathscr{N}(\rho_t)  = \frac{\norm{\rho^{T_{B}}_t}_{1}-1}{2}.
\end{equation}
After differentiating the above equation with respect to time $t$, we obtain
\begin{align}
 \frac{{\rm d} }{{\rm d}t}\cal{N}(\rho_t) = 
 \frac{1}{4}\tr\left({\left|\rho^{T_{B}}_t\right|^{-1}\left((\dot{\rho}^{T_{B}}_t)^{\dagger}\rho^{T_{B}}_t+(\rho^{T_{B}}_t)^{\dagger}\dot{\rho}^{T_{B}}_t\right) }\right).
 \end{align}
Let us now consider the absolute value of the above equation and apply the property $|\tr(A)|\leq \tr|A|$. We then obtain the following inequality
\begin{equation}
   \left|\frac{{\rm d} }{{\rm d}t}\mathscr{N}(\rho_t)\right| \leq \frac{1}{4}\tr\left(\left|\rho^{T_{B}}_t\right|^{-1}{\left|(\dot{\rho}^{T_{B}}_t)^{\dagger}\rho^{T_{B}}_t+(\rho^{T_{B}}_t)^{\dagger}\dot{\rho}^{T_{B}}_t\right|}\right),
\end{equation}
where $({\rho}^{T_{B}}_t)^{\dagger}={\rho}^{T_{B}}_t$ (see Appendix \ref{Herm}). Let us use the triangular inequality for further simplification. We then obtain
\begin{equation}\label{eq:ent-bound}
   \left|\frac{{\rm d} }{{\rm d}t}\mathscr{N}(\rho_t)\right| \leq \frac{1}{2}\tr\left({\left|\dot{\rho}^{T_{B}}_t\right|}\right) = \frac{1}{2}\norm{\dot{\rho}^{T_{B}}_t}_{1}.
\end{equation}
The above inequality~\eqref{eq:ent-bound} is the upper bound on that the rate of change of the negativity of the quantum system evolving under given dynamics. After integrating the above equation with respect to time $t$, we obtain 
\begin{equation}
  \int_{0}^{T} {\rm d}t\left|\frac{{\rm d} }{{\rm d}t}\mathscr{N}(\rho_t)\right| \leq \frac{1}{2}\int_{0}^{T}{\rm d}t\norm{\dot{\rho}^{T_{B}}_t}_{1} = \frac{1}{2}\int_{0}^{T}{\rm d}t\norm{\cal{L}_t(\rho^{T_{B}})}_{1}.
\end{equation}
From the above inequality, we get the desired bound:
\begin{equation}
T \geq   \frac{2\left| \mathscr{N}(\rho_{T})-\mathscr{N}(\rho_{0})\right|}{ \Lambda^{N}_{T}}.
\end{equation}
\end{proof}

\section{$T_{\rm NSL}$ and $T_{\rm CSL}$ for nonlocal Hamiltonian}\label{Negativity_and_Concurrence_For_Non-Local_Hamiltonian}
Using Eq.~\eqref{equ:negativity_for_non-local_dynamics} and Eq~\eqref{equ:concurrence_for_non-local_dynamics} we can check that the negativity and concurrence are periodic functions in time for given unitary dynamics with nonlocal Hamiltonian $\widetilde{H}^+_{AB}$ and period can be altered by changing the value of $\theta=\mu_{x}-\mu_{y}$. It is clear from Fig.~\ref{fig:theta_equals_0.5} and Fig.~\ref{fig:theta_equals_1} that the time period of negativity and square of concurrence is higher for small values of $\theta$. From Fig.~\ref{fig:speed_linit_plot_for_theta_equal_0,5} and Fig.~\ref{fig:speed_linit_plot_for_theta_equal_1} it is observed that our corresponding speed limits are relatively tighter for small values of $\theta$.  
\begin{figure}
     \centering
     \begin{subfigure}[b]{0.37\textwidth}
         \centering
         \includegraphics[width=\textwidth]{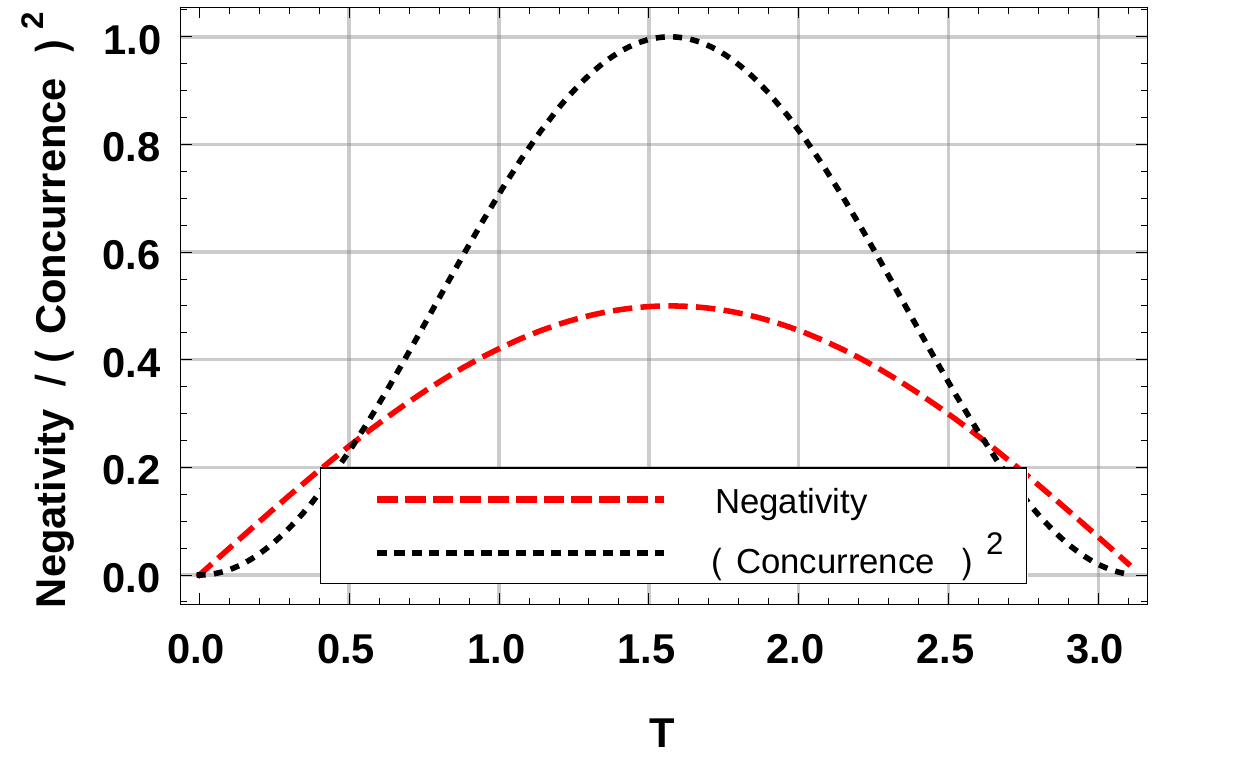}
         \caption{negativity and square of concurrence vs $T$ with $\theta=0.5$ and $p=0$.}
         \label{fig:theta_equals_0.5}
     \end{subfigure}
     \hfill
     \begin{subfigure}[b]{0.40\textwidth}
         \centering
         \includegraphics[width=\textwidth]{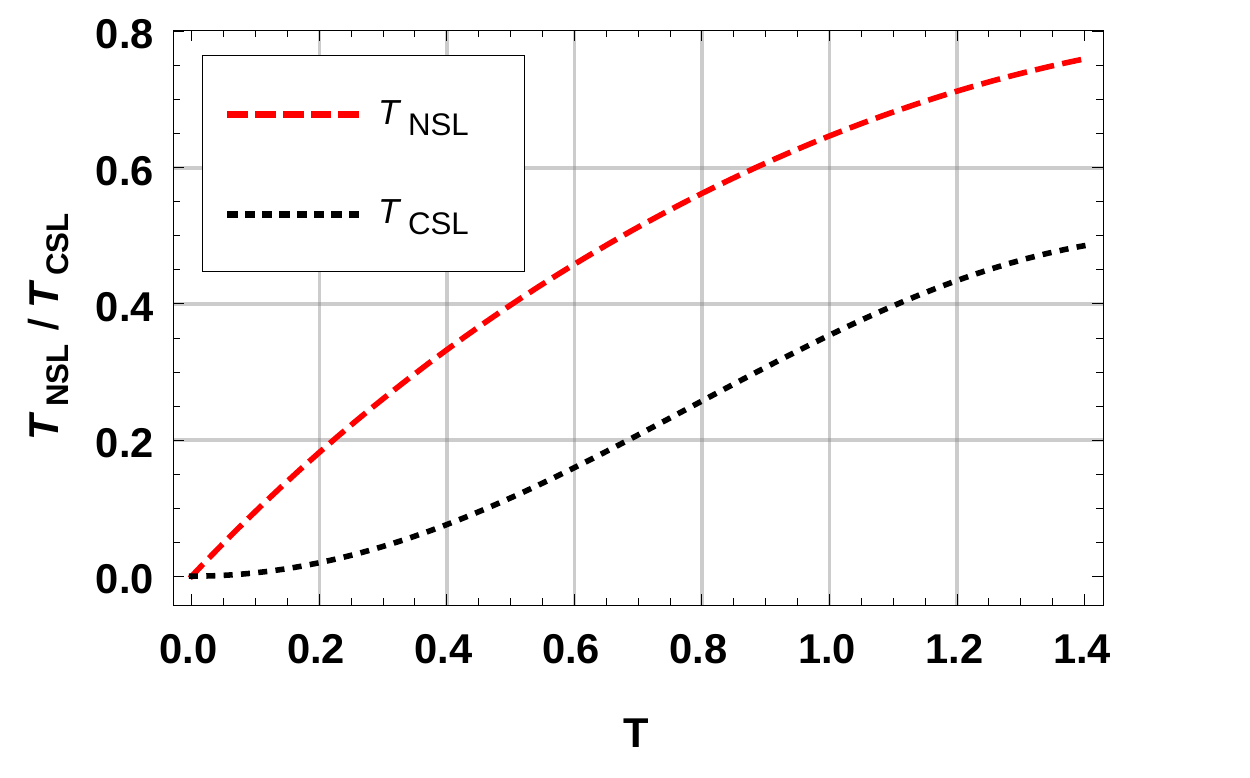}
         \caption{ $T_{\rm NSL}$/$T_{\rm CSL}$ vs $T$ with $\theta=0.5$ and $p=0$.}
         \label{fig:speed_linit_plot_for_theta_equal_0,5}
     \end{subfigure}
     \hfill
 \begin{subfigure}[b]{0.37\textwidth}
         \centering
         \includegraphics[width=\textwidth]{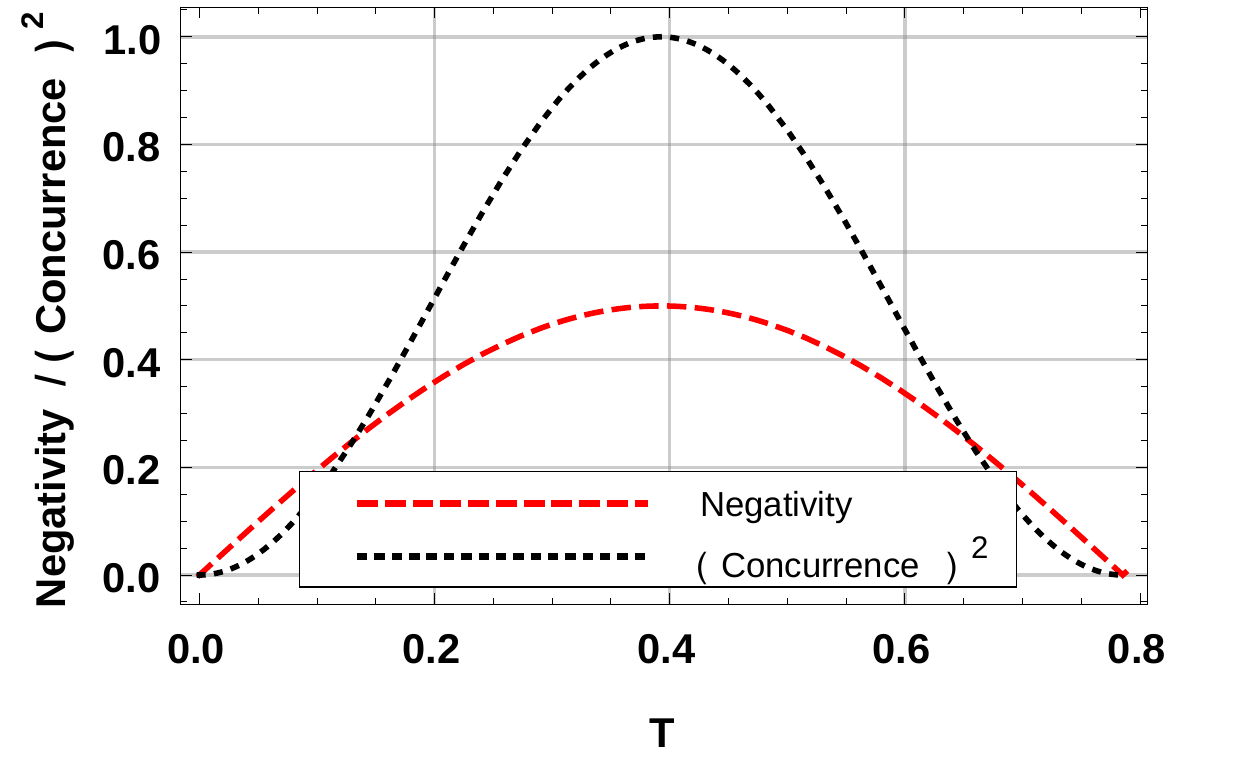}
         \caption{negativity and square of concurrence vs $T$ with $\theta=2$ and $p=0$.}
         \label{fig:theta_equals_1}
     \end{subfigure}
     \hfill
         \begin{subfigure}[b]{0.40\textwidth}
         \centering
         \includegraphics[width=\textwidth]{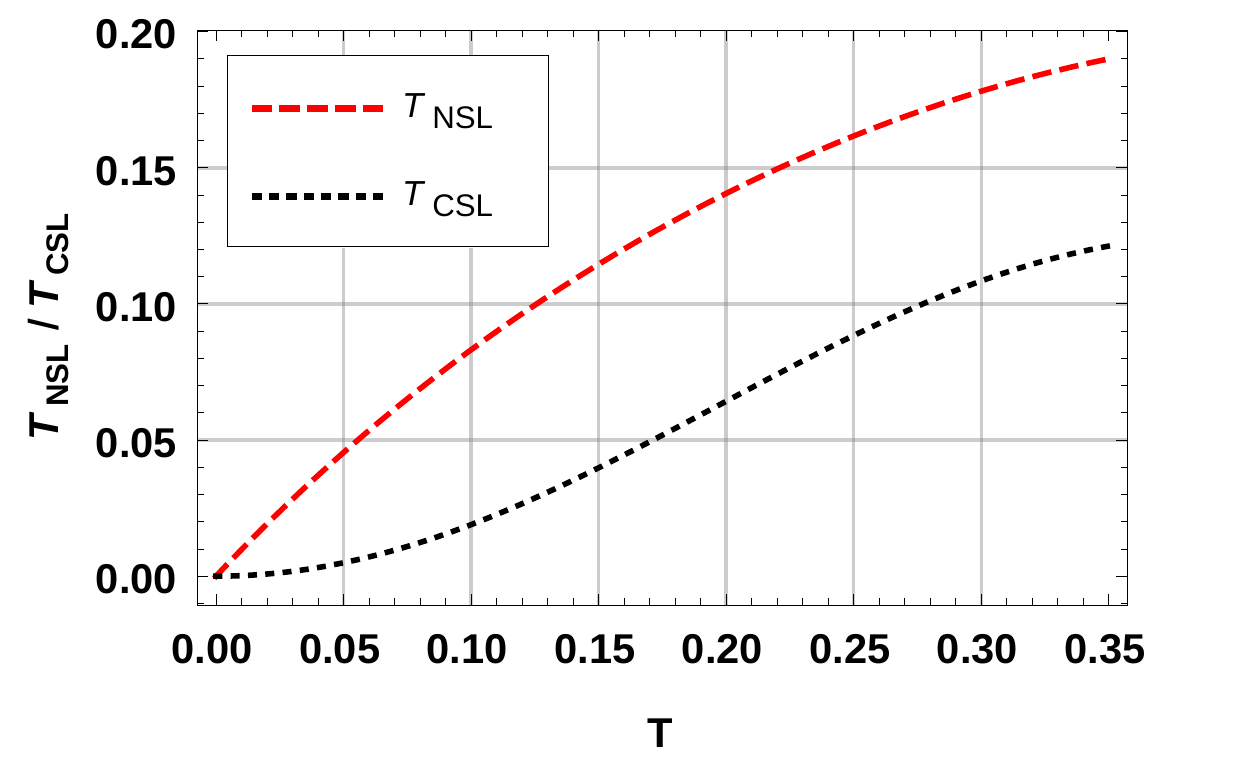}
         \caption{$T_{\rm NSL}$/$T_{\rm CSL}$ vs $T$ with $\theta=2$ and $p=0$.}
         \label{fig:speed_linit_plot_for_theta_equal_1}
     \end{subfigure}
        \caption{Here we depict square of the concurrence and negativity vs $T$ in Fig.~\ref{fig:theta_equals_0.5} and Fig.~ \ref{fig:theta_equals_1}. In Fig.~\ref{fig:speed_linit_plot_for_theta_equal_0,5} and Fig.~\ref{fig:speed_linit_plot_for_theta_equal_1}, we depict $T_{\rm NSL}$ and $T_{\rm CSL}$ vs $T$ with $\theta = \{0.5,2\}$ and initial state with $p=0$.}
        \label{fig:three graphs}
\end{figure}

\section{Proof of Theorem \ref{thm-t-independent}}\label{proof_of_tmh_2}
\begin{proof}
 The square of the concurrence of time evolved bipartite pure state $\psi_{t}$ given by 
\begin{equation}
    \mathscr{C}^{2}(\psi_t) = {\tr(\psi_t\mathcal{R}({\psi^{*}_t}))},
\end{equation}
After differentiating the above equation with respect to time $t$, we obtain
\begin{align}
    \frac{\rm d}{{\rm d} t}\mathscr{C}^{2}(\psi_t) =  \tr\big(\Dot{\psi_t}\mathcal{R}({\psi^{*}_t})\big) + \tr\Big(\psi_t\mathcal{R}(\dot{\psi}^{*}_t)\Big),
\end{align}
where $\mathcal{R}(\cdot)= \sigma_{y}\otimes \sigma_{y} (\cdot) \sigma_{y}\otimes \sigma_{y}$.
Let us now consider the absolute value of the above equation and apply triangular inequality $|A+B|\leq |A|+|B|$. We then obtain the following inequality
\begin{equation}
    \left|\frac{\rm d}{{\rm d}t}\mathscr{C}^{2}(\psi_t)\right| \leq \left|\tr\big(\Dot{\psi_t}\mathcal{R}({\psi^{*}_t})\big)\right|+
    \left|\tr\big(\mathcal{R}(\psi_t)\Dot{\psi_t}^{*}\big)\right|.
\end{equation}
Using Liouville-von Neumann equation $\Dot{\psi}_t=\frac{\iota}{\hbar}[\psi_t,H_{t}]$ for unitary dynamics, the above inequality can be expressed as 
\begin{align}
    \left|\frac{\rm d}{{\rm d}t}\mathscr{C}^{2}(\psi_t)\right| &\leq \frac{1}{\hbar}[\left|\tr([\psi_t,H_{t}]\mathcal{R}({\psi^{*}_t}))\right|\nonumber\\
    &\hspace{0.35cm}+
    \left|\tr(\mathcal{R}(\psi_t)([\psi_t,H_{t}])^{*})\right|].
\end{align}
By applying triangular inequality once more on the above inequality, we then obtain
\begin{align}
    \left|\frac{\rm d}{{\rm d}t}\mathscr{C}^{2}(\psi_t)\right| &\leq \frac{1}{\hbar}[\left|\tr(\psi_t H_{t}\mathcal{R}({\psi^{*}_t}))\right|+\left|\tr(H_{t}\psi_t \mathcal{R}({\psi^{*}_t}))\right|\nonumber\\
    & +\left|\tr(\mathcal{R}(\psi_t)(\psi_t H_{t})^{*})\right|+\left|\tr(\mathcal{R}(\psi_t)(H_{t}\psi_t )^{*})\right|].
\end{align}
Let us apply the Cauchy--Schwarz inequality $|{\rm tr}(AB)|\leq\sqrt{{\rm tr}(A^{\dagger}A){\rm tr}(B^{\dagger}B)}$. We then obtain the following inequality
\begin{align} 
    \left|\frac{\rm d}{{\rm d}t}\mathscr{C}^{2}(\psi_t)\right| &\leq \frac{2}{\hbar}[\sqrt{\tr(\mathcal{R}({\psi^{*}_t})^{\dagger}\mathcal{R}({\psi^{*}_t}))} \sqrt{\tr(\psi_t H_{t}^2)} \nonumber\\
    & +\sqrt{\tr({\mathcal{R}(\psi_t)}^{\dagger}\mathcal{R}(\psi_t))} \sqrt{\tr(\psi_t H_{t}^2)^{*}}.
\end{align}
Since $\psi_t$ is pure state, it implies that $\sqrt{\tr(\mathcal{R}({\psi^{*}_t})^{\dagger}\mathcal{R}({\psi^{*}_t})} = 1$, $\sqrt{\tr(\mathcal{R}(\psi_t)^{\dagger}\mathcal{R}(\psi_t))} = 1$ and $\tr(\psi_t H^{2})^{*} = \tr(\psi_t H^{2})$.
Therefore, we can rewrite the above equation as 
\begin{equation}\label{eq:con-bound}
    \left|\frac{\rm d}{{\rm d}t}\mathscr{C}^{2}(\psi_t)\right| \leq \frac{4}{\hbar}\sqrt{\tr(\psi_t H_{t}^2)}.
\end{equation}
The above inequality~\eqref{eq:con-bound} is the upper bound on that the rate of change of square of the concurrence of the quantum system evolving under unitary dynamics. After integrating the above equation with respect to time $t$, we obtain 
\begin{equation}
   \int_{0}^{T} \left|\frac{\rm d}{{\rm d}t}\mathscr{C}^{2}(\psi_t)\right| \leq \frac{4}{\hbar} \int_{0}^{T}\sqrt{\tr(\psi_t H_{t}^2)}{\rm d}t.
\end{equation}
From the above inequality, we get the desired bound:
\begin{equation}
T \geq  \frac{\hbar}{4}\frac{\left|{ \mathscr{C}^2(\psi_T)}-{ \mathscr{C}^2(\psi_0)}\right|}{\Lambda^{C}_T}.
\end{equation}
\end{proof}

\section{Alternative speed limit on Bell-CHSH correlations}\label{app:chsh}
In this section, we derive speed limits for Bell-CHSH observable whose dynamics governed by separable map. The separable dynamics degrade the Bell-CHSH correlation.
Expectation value of Bell-CHSH observable at time $t$ is given by
\begin{equation}
    \tr(\phi_{t}(\rho)\cal{B})= \tr(\phi_{t}(\rho) a_1\otimes b_1  ) + \tr(\phi_{t}(\rho) a_2\otimes b_2  ),
\end{equation}
where $a_1=\hat{a}.\vec{\sigma}$, $a_2=\hat{a'}.\vec{\sigma}$,  $b_1=(\hat{b}+\hat{b'}).\vec{\sigma}$ and $b_2=(\hat{b}-\hat{b'}).\vec{\sigma}$.
The above equation can be rewritten as
\begin{equation}
    \tr(\rho\phi^{\dagger}_{t}(\cal{B}))= \tr(\rho\phi^{\dagger}_{t}( a_1\otimes b_1)  ) + \tr(\rho\phi^{\dagger}_{t}( a_2\otimes b_2)  ).
\end{equation}
Let us differentiate  above equation with respect to time t, we then obtain
\begin{equation}
     \frac{\rm d}{{\rm d}t}\tr(\rho\cal{B}_t)= \tr(\rho \cal{L}_{A}(a_{1t})\otimes \cal{L}_{B}(b_{1t}))   + \tr(\rho \cal{L}_{A}(a_{2t})\otimes \cal{L}_{B}(b_{2t})),
\end{equation}
where $\cal{L}_{AB}=\cal{L}_A\otimes \operatorname{id}_B + \operatorname{id}_A\otimes \cal{L}_B$, $\operatorname{id}_{A}$ and $\operatorname{id}_{B}$ are the identity super-operators corresponding to system $A$ and $B$ respectively. Let us take the absolute value of the above equation and applying the triangular inequality, we get
\begin{align}
     \left|\frac{\rm d}{{\rm d}t}\tr(\rho\cal{B}_t)\right|\leq & \left|\tr(\rho \cal{L}_{A}(a_{1t})\otimes \cal{L}_{B}(b_{1t}))\right|  \nonumber \\
     & + \left|\tr(\rho \cal{L}_{A}(a_{2t})\otimes \cal{L}_{B}(b_{2t}))\right|.
\end{align}
Let us apply the Cauchy--Schwarz inequality, we then obtain the following inequality
\begin{align}
     \left|\frac{\rm d}{{\rm d}t}\tr(\rho\cal{B}_t)\right|\leq & \sqrt{\tr(\rho^2)} (\norm{\cal{L}_{A}(a_{1t})\otimes \cal{L}_{B}(b_{1t})}_{2} \nonumber \\  & + \norm{ \cal{L}_{A}(a_{2t})\otimes \cal{L}_{B}(b_{2t})}_{2}).
\end{align}
After integrating the above equation with respect to time $t$, we then obtain the following bound
\begin{equation}
 T \geq  T_{\rm BQSL}= \frac{| \langle\cal{B}_{T}\rangle_{\rho}-\langle\cal{B}_{0}\rangle_{\rho}|}{\sqrt{\tr(\rho^2)}\Lambda^{B}_{T}},
 \end{equation}
 where $\Lambda^{B}_{T}=\int_{0}^{T}{\rm d}t (\norm{\cal{L}_{A}(a_{1t})\otimes \cal{L}_{B}(b_{1t})}_{2}   + \norm{ \cal{L}_{A}(a_{2t})\otimes \cal{L}_{B}(b_{2t})}_{2})$.

 The above abound describes how fast Bell-CHSH correlation degrades under separable dynamics.

\section{Proof of speed limit on quantum mutual information~\eqref{eq:theorem}}\label{MI}
\begin{proof}
Consider $\omega^{AB}_0:=\rho^{A}_0\otimes\rho^{B}_0$ to be an initial state (at $t=0$) of the quantum system before interaction or dynamics begin. The relative entropy of the time evolved quantum state $\omega_{t}$ with respect to $\omega_{0}$ is given as
\begin{equation}
     D\left(\omega_{t}\Vert\omega_{0}\right)= \tr
     \{\omega_{t}\ln{\omega_t}-\omega_{t}\ln{\omega_0}\},
\end{equation}
where $\operatorname{supp}(\omega_t)\subseteq\operatorname{supp}(\omega_0)$ for all $t\geq 0$. After differentiating above equation with respect to time $t$, we then obtain~\cite{Spohn1978} (see also Theorem 1 of Ref.\cite{Das2018})
\begin{equation}
   \frac{{\rm d} }{{\rm d}t}  D\left(\omega_t\Vert\omega_{0}\right)= \tr\left(\cal{L}_t\left(\omega_t\right)\left({\Pi}_{t}   \ln\omega_t-\Pi_{0}\ln{\omega_0}\right)\right),
\end{equation}
where $\Pi_t$ denotes the projection onto the support of $\omega_t$. Taking the absolute value of the terms in the above equation and applying the Cauchy--Schwarz inequality, we get
\begin{equation}
   \left|\frac{{\rm d} }{{\rm d}t}  D\left(\omega_t\Vert\omega_{0}\right)\right|\leq \norm{\mathcal{L}_{t}({\omega_t})}_{\rm 2} \norm{\Pi_{t}\ln\omega_t-\Pi_{0}\ln{\omega_0}}_{\rm 2}.
\end{equation}
As $\ln\omega_t$ is defined on the support of $\omega_t$, we have
\begin{equation}
   \abs{\frac{{\rm d} }{{\rm d}t}  D\left(\omega_t\Vert\omega_{0}\right)}\leq \norm{\mathcal{L}_{t}({\omega_t})}_{\rm 2} \norm{\ln\omega_t-\ln{\omega_0}}_{\rm 2}.
\end{equation}
 After integrating above equation with respect to time $t$, we obtain
\begin{equation}
  T\geq  \frac{D\left(\omega_T\Vert\omega_{0}\right)}{ {\Lambda}_{T}^{M}  },
\end{equation}
where $\Lambda^{M}_{T}:=\frac{1}{T}\int_{0}^{T}{\rm d}t\norm{\mathcal{L}_{t}({\omega_t})}_{2} \norm{\ln\omega_t-\ln{\omega_0}}_{2}$. Let use the fact $D\left(\omega_T||\omega_{0}\right)\geq \min \limits_{\{\rho^{A}_0,  \rho^{B}_0\}}D\left(\omega_T||\omega_{0}\right)={I}(A;B)_{\omega_T}$, we then obtain
\begin{equation}
  T\geq  \frac{{I}\left(A;B\right)_{\omega_T}}{ {{\Lambda}}_{T}^{M}  }.
\end{equation}
\end{proof}

\section{Proof of speed limit on the entropy~\eqref{EEbound}}\label{entropy}
Consider the state of a quantum system $A$ evolving under quantum dynamics at time $t$ is given by $\rho_t$. Then the entropy of the quantum system at time $t$ is given by
\begin{equation}
   S(\rho_t)= -{\rm tr}\{\rho_t \ln \rho_t\}.
\end{equation}
It is standard convention that $0\ln 0=0$.
After differentiating the above equation with respect to time $t$, we obtain~\cite{Das2018,Spohn1978}
\begin{equation}
   \frac{{\rm d}}{{\rm d}t} S(\rho_t^{ A})= -{\rm tr}\{\dot{\rho}_t^{ A}\ln \rho_t^{ A}\}=-{\rm tr}\{\mathcal{L}_{t}({\rho_t^{ A}})\ln \rho_t^{ A}\}.
\end{equation}
Let us now consider the absolute value of the above equation and apply the Cauchy--Schwarz inequality. We then obtain the following inequality
\begin{align}\label{eq:entbound}
   \left|\frac{{\rm d} }{{\rm d}t}S(\rho_t^{ A})\right|&= \left|{\rm tr}\{ \mathcal{L}_{t}({\rho_t^{A}}) \ln \rho_t^{A}\}\right| \nonumber\\
   &\leq \norm{\mathcal{L}_{t}({\rho_t^{ A}})}_{2} \norm{ \ln \rho_t^{ A}}_{2},
\end{align}
where $\ln\rho_t$ is defined on the support of $\rho_t$.
The above inequality~\eqref{eq:entbound} is the upper bound on that the rate of change of the entropy of the quantum system evolving under given dynamics. After integrating above equation with respect to time $t$, we obtain 
\begin{equation}\label{MER}
  \int_{0}^{T} {\rm d}t\left|\frac{{\rm d} }{{\rm d}t}S(\rho_t^{ A})\right| \leq \int_{0}^{T} \norm{\mathcal{L}_{t}({\rho_t^{ A}})}_{\rm 2} \norm{\ln \rho_t^{ A}}_{\rm 2}   {\rm d}t.
\end{equation}
From the above inequality, we get the desired bound:
\begin{equation}
  T\geq  \frac{\left|S(\rho_T^{ A}) -S(\rho_0^{ A})\right|}{ {\Lambda}_{T}},
\end{equation}
where $\Lambda_{T}=\frac{1}{T}\int_{0}^{T}{\rm d}t \norm{\mathcal{L}_{t}({\rho_t^{ A}})}_{\rm 2} \norm{\ln \rho_t^{ A}}_{2}$ is the evolution speed of entropy.

Note that if we further apply the Cauchy-–Schwarz inequality in Eq.~\eqref{MER} as done in the proof argument of Theorem~1 of Ref.~\cite{Mohan2022}, we get comparatively weaker bound, see Eq.~(6) of Ref.~\cite{Mohan2022}. This observation of using the Cauchy--Schwarz inequality only once also allows for obtaining tighter bounds for other informational measures discussed in Ref.~\cite{Mohan2022}.

\section{Amplitude Damping Process}\label{apmlitude_damping_process}

We now consider dynamical processes describable by amplitude-damping channel. The Lindbladian operators for amplitude damping process are given as $L^{A}_{amp} = \sqrt{\frac{\gamma^A}{2}} \sigma^A_{-}\otimes\mathbbm{1}_{B}$ and $L^B_{amp} = \sqrt{\frac{\gamma^B}{2}}\mathbbm{1}_{A} \otimes\sigma^B_{-} $, where $\sigma^A_{-} \equiv \ket{1}\bra{0}_{A}$ and $\sigma^B_{-} \equiv \ket{1}\bra{0}_{B}$ are spin lowering operators and $\gamma^A,\gamma^B\in\mathbbm{R}$ denote the strength of amplitude damping. The LGKS master equation governs the time evolution of bipartite state  $\rho_{t}$ in Schr\"odinger picture and Bell-CHSH observable $\cal{B}_{t}$ in Heisenberg picture:
 \begin{align}
      \frac{\rm d}{{\rm d}t}\rho_{t} &= \frac{\gamma^A}{2} (2\sigma^A_{-}\otimes \mathbbm{1}_{B}(\rho_{t})\sigma^A_{+}\otimes \mathbbm{1}_{B} -\{\sigma^A_{+}\sigma^A_{-}\otimes \mathbbm{1}_{B},\rho^{AB}_{t}\})\nonumber\\
  & \hspace{0.35cm}+\frac{\gamma^B}{2} (2 \mathbbm{1}_{A}\otimes\sigma^B_{-}(\rho_{t})\mathbbm{1}_{A}\otimes\sigma^B_{+} -\{\mathbbm{1}_{A}\otimes \sigma^B_{+}\sigma^B_{-},\rho^{AB}_{t}\}),\\
      \frac{\rm d}{{\rm d}t}\cal{B}_{t} &= \frac{\gamma^A}{2} (2\sigma^A_{+}\otimes \mathbbm{1}_{B}(\cal{B}_{t})\sigma^A_{-}\otimes \mathbbm{1}_{B} -\{\sigma^A_{+}\sigma^A_{-}\otimes \mathbbm{1}_{B},\cal{B}_{t}\})\nonumber\\
  & \hspace{0.35cm}+\frac{\gamma^B}{2} (2 \mathbbm{1}_{A}\otimes\sigma^B_{+}(\cal{B}_{t})\mathbbm{1}_{A}\otimes\sigma^B_{-} -\{\mathbbm{1}_{A}\otimes \sigma^A_{+}\sigma^A_{-},\cal{B}_{t}\}),
 \end{align}
 where $\sigma^A_{+} \equiv \ket{0}\bra{1}_{A}$ and $\sigma^B_{+} \equiv \ket{0}\bra{1}_{B}$ are spin raising operator corresponding to $A$ and $B$ respectively. The solutions of the above equations for amplitude damping process is given by 
  \begin{align}
      \rho_{t} &= p {\rm e}^{-2 \gamma t}  \ket{00}\bra{00}+ \sqrt{(1-p) p} {\rm e}^{-2 \gamma  t} \left(\ket{00}\bra{11}+\ket{11}\bra{00}\right)\nonumber \\
      &\hspace{0.35cm}+p {\rm e}^{-2 \gamma  t} \left({\rm e}^{\gamma  t}-1\right)\left(\ket{01}\bra{01}+\ket{10}\bra{10}\right)\nonumber \\
      & \hspace{0.35cm}+ (1-p+p {\rm e}^{-2 \gamma  t} \left({\rm e}^{\gamma  t}-1\right)^2)
      \ket{11}\bra{11},\\
      \cal{B}_{t} &={\rm e}^{-2 \gamma  t} \left(8 \cos (\eta )-8 \cos (\eta ) {\rm e}^{\gamma  t}+2 \cos (\eta ) {\rm e}^{2 \gamma  t}\right)\ket{00}\bra{00}\nonumber \\
      &\hspace{0.35cm}+{\rm e}^{-\gamma t} \left(2 \cos (\eta ) {\rm e}^{\gamma  t}-4 \cos (\eta )\right)\left(\ket{01}\bra{01}+\ket{10}\bra{10}\right) \nonumber \\
     & \hspace{0.35cm} +2 \cos (\eta )\ket{11}\bra{11}+ 2 \sin (\eta ) {\rm e}^{-\gamma  t}\left(\ket{00}\bra{11}+\ket{10}\bra{01}\right. \nonumber\\
     & \hspace{0.35cm}+\left.\ket{01}\bra{10}+\ket{11}\bra{00}\right),
  \end{align}
  where we assumed $\gamma^A = \gamma^B = \gamma$. To estimate bounds on the negativity~\eqref{eq:theorem-1} and Bell-CHSH observable~\eqref{eq:theorem-3}, we need the following quantities:
\begin{align}
 \left | \mathscr{N}(\rho_{T})-\mathscr{N}(\rho_{0})\right| &=\left| -{\rm e}^{-2 \gamma  t}\left(-\sqrt{(p-p^2 )}{\rm e}^{ \gamma  t}+p {\rm e}^{\gamma  t}-p\right)\nonumber\right. \\
  &\hspace{0.4cm}\left.-\sqrt{p-p^2}\right|,\\
 \norm{\cal{L}_t(\rho_{t}^{T_{B}})}_{1} &= {\rm e}^{-2 \gamma  t}\bigg({\rm e}^{2\gamma t}\left(\sqrt{c-2\sqrt{d}}+ \sqrt{c+2\sqrt{d}} \right)\nonumber\\
 & \hspace{0.35cm}  +2 c' {\rm e}^{\gamma t}+2  \gamma p \bigg),
 \end{align}
 \begin{align}
  \left| \langle\cal{B}_{T}\rangle_{\rho_{0}}-\langle\cal{B}_{0}\rangle_{\rho_{0}}\right| &=\left| {\rm e}^{\gamma t} \sin (\eta )\sqrt{1-p}+2 \cos (\eta ) \sqrt{p}\right|\nonumber \\
  &\hspace{0.35cm}4 \sqrt{p} {\rm e}^{-2 \gamma  t} \left({\rm e}^{\gamma  t}-1\right),\\
     \min\left\{\Lambda_{T}^{\infty},\Lambda_{T}^{1},\Lambda_{T}^{2}\right\}& =2 \gamma  {\rm e}^{-2 \gamma  t}\bigg(\sqrt{1-d'}+\sqrt{d'+1}\nonumber\\
     &\hspace{0.35cm}+\sqrt{2-2\sqrt{\left(1-d'^2\right)}}\bigg)\label{equ:singular_value_of_bell_operator_in_amplitude_damping_process},
\end{align}
where $c$, $d$, $c'$, and $d'$ are functions of $\gamma$, $t$, and $p$ such that $c=\gamma ^2 p {\rm e}^{-4 \gamma  t} \left(-4 p {\rm e}^{\gamma  t}+4 p+{\rm e}^{2 \gamma  t}\right)$, $d= \gamma ^4 \left({\rm e}^{-6 \gamma  t}\right)(1-p) p^3  \left({\rm e}^{\gamma  t}-2\right)^2$, $c'= \sqrt{\gamma ^2 p^2 {\rm e}^{-2 \gamma  t} \left({\rm e}^{\gamma  t}-1\right)^2}$, and $d'=\sin(2\eta)$.

\begin{figure}
     \centering
     \begin{subfigure}[b]{0.44\textwidth}
         \centering
         \includegraphics[width=\textwidth]{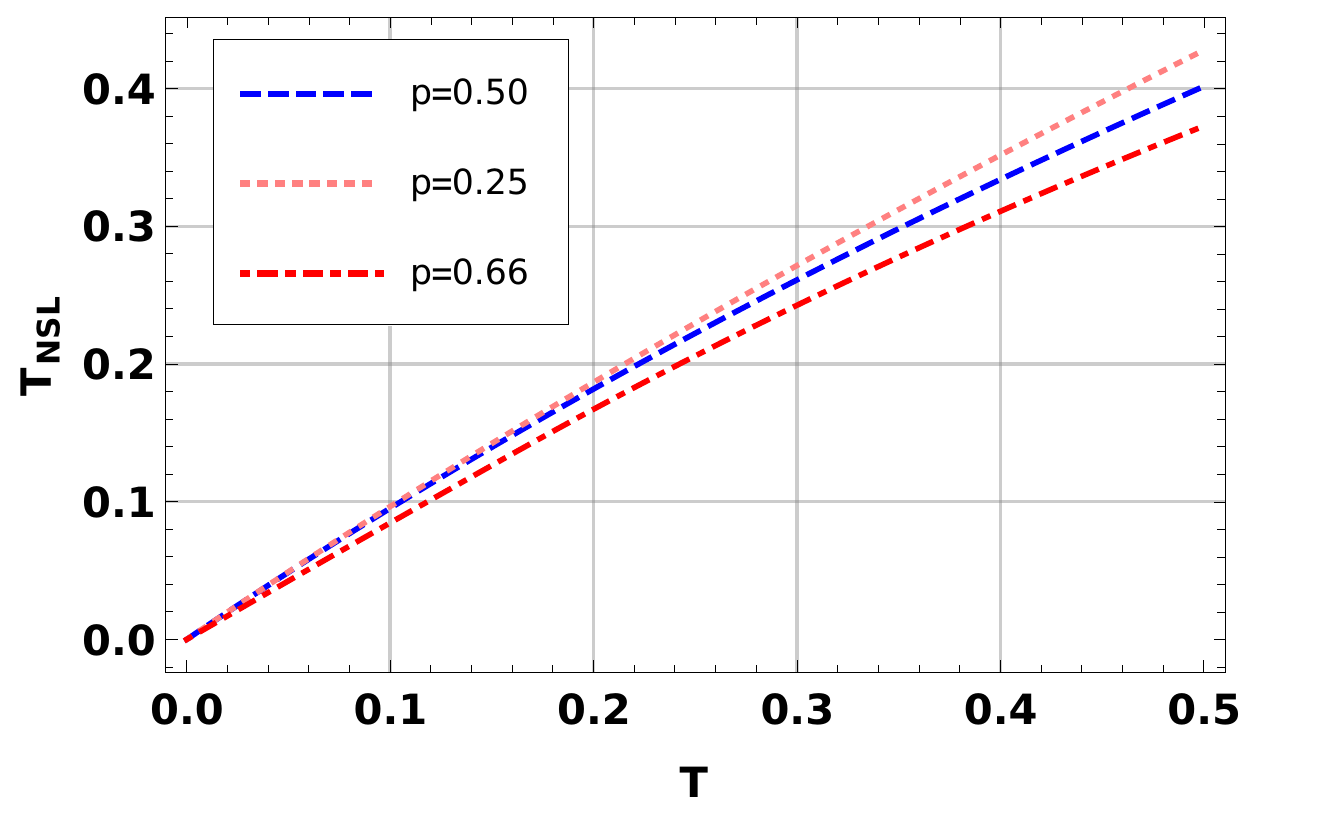}
         \caption{$T_{\rm NSL}$ vs $T$ with initial state parameter $p=\{0.25,0.50,0.66\}$.}
         \label{fig:Amplitude-Damping:Neg}
     \end{subfigure}
     \hfill
     \begin{subfigure}[b]{0.44\textwidth}
         \centering
         \includegraphics[width=\textwidth]{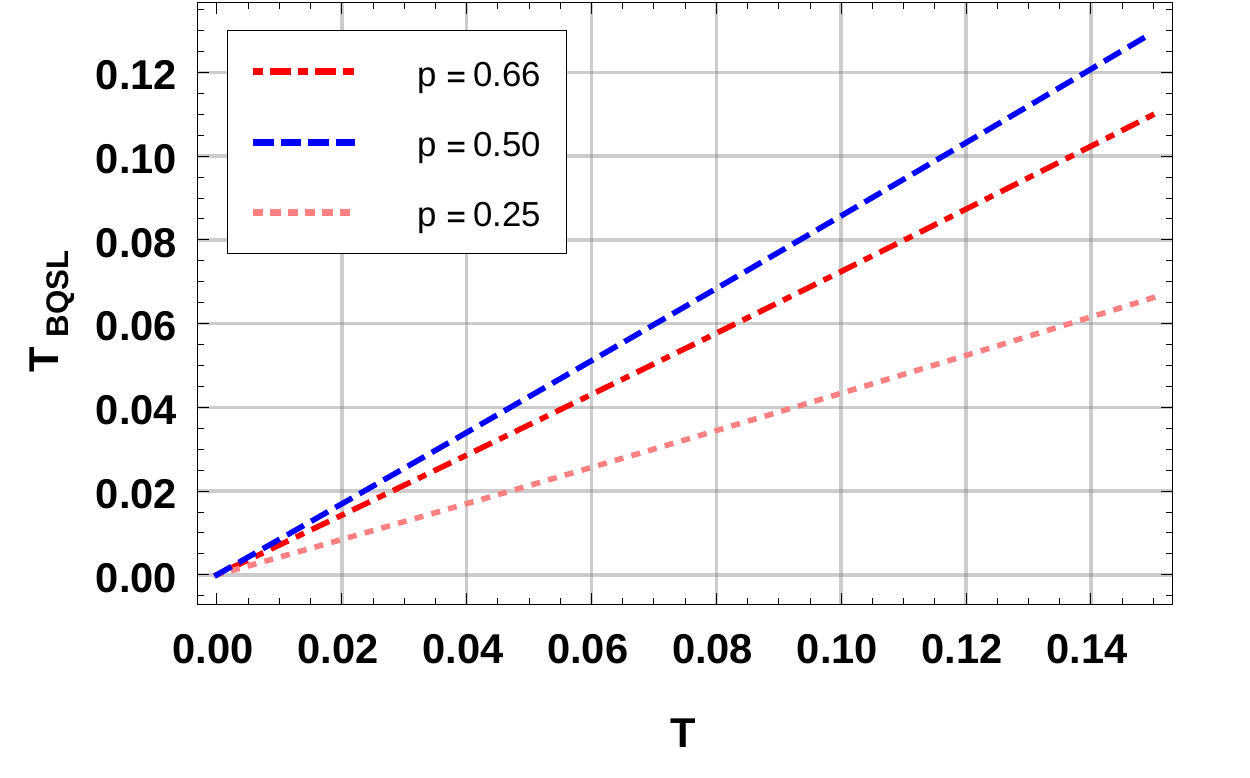}
         \caption{$T_{\rm BQSL}$ vs $T$ with initial state parameter $p=\{0.25,0.50,0.66\}$.}
         \label{fig:Amplitude-Damping:Bell}
     \end{subfigure}
  \caption{For amplitude damping process with $\gamma = 1$, we depict $T_{\rm NSL}$~\eqref{eq:theorem-1} vs $T$ in (a) and $T_{\rm BQSL}$~\eqref{eq:theorem-3} vs $T$ in (b).}
    \label{fig:Amplitude-Damping}
\end{figure}  
Note that Eq.~\eqref{equ:singular_value_of_bell_operator_in_amplitude_damping_process} is only valid in the interval  $t \in [0,0.15]$.

In Fig.~\ref{fig:Amplitude-Damping:Neg}, we plot $T_{\rm NSL}$~\eqref{eq:theorem-1} vs $T$ where $T\in [0,0.50]$ for amplitude damping process when $\gamma = 1$ and $p\in\{0.25,0.50,0.66\}$. We find that the negativity degrades faster for $p\in\{0.50,0.66\}$ in comparison to $p=0.25$. We note that the speed limit on the negativity~\eqref{eq:theorem-1} is not tight for given amplitude damping process.

In Fig.~\ref{fig:Amplitude-Damping:Bell}, we plot $T_{\rm BQSL}$~\eqref{eq:theorem-3} vs $T$ where $T\in [0,0.50]$ for amplitude damping process when $\gamma = 1$ and $p\in\{0.25,0.50,0.66\}$. We find that the Bell-CHSH correlation degrades faster for $p\in\{0.25,0.66\}$ (non-maximally entangled state) in comparison to $p=0.50$ (maximally entangled state). We note that the bound~\eqref{eq:theorem-3} is tight and attainable for amplitude damping process when $\gamma = 1$ and $p=0.50$.

\bibliography{main}

\end{document}